\documentclass[10pt, pra, a4paper, notitlepage, raggedbottom, amsmath, amssymb, floatfix]{revtex4-1}

\usepackage[utf8]{inputenc}
\usepackage{kpfonts}

\usepackage{graphicx}

\usepackage{amsthm}
\usepackage{mathtools}  
\usepackage{bm}         
\usepackage{dsfont}	
\usepackage{ytableau}   
\usepackage{centernot}  
\usepackage{enumitem}   

\usepackage{placeins}   
\usepackage{xcolor}
\usepackage{microtype}  

\usepackage{longtable}  
\usepackage{booktabs}   
\usepackage{dcolumn} 	

\usepackage[breaklinks]{hyperref}
\hypersetup{
    colorlinks,
    linkcolor={red!40!black},
    citecolor={blue!50!black},
    urlcolor={blue!80!black}
}

\newtheorem{theorem}{Theorem}
\newtheorem{proposition}[theorem]{Proposition}
\newtheorem{observation}[theorem]{Observation}
\newtheorem{lemma}[theorem]{Lemma}
\newtheorem{corollary}[theorem]{Corollary}

\newtheorem{example}[theorem]{Example}

\theoremstyle{remark}
\newtheorem*{remark}{Remark}

\theoremstyle{definition}
\newtheorem{definition}[theorem]{Definition}

\DeclareMathOperator{\tr}{tr}

\DeclareMathOperator{\sgn}{sgn}

\newcommand{\ot}[0]{\otimes}
\newcommand{\nn}[0]{\nonumber}

\newcommand{\BB}[0]{\mathcal{B}}
\newcommand{\CC}[0]{\mathcal{C}}

\newcommand{\HH}[0]{\mathcal{H}}

\newcommand{\JJ}[0]{\mathcal{J}}

\newcommand{\MM}[0]{\mathcal{M}}

\newcommand{\PP}[0]{\mathcal{P}}

\newcommand{\RR}[0]{\mathcal{R}}
\renewcommand{\SS}[0]{\mathcal{S}}
\newcommand{\TT}[0]{\mathcal{T}}
\newcommand{\UU}{\mathcal{U}}

\newcommand{\WW}{\mathcal{W}}

\newcommand{\N}{\mathds{N}}
\newcommand{\R}{\mathds{R}}
\newcommand{\C}{\mathds{C}}

\newcommand{\one}[0]{\mathds{1}}

\newcommand{\id}[0]{\operatorname{id}}

\newcommand{\bra}[1]{\mathinner{\langle #1|}}	
\newcommand{\ket}[1]{\mathinner{|#1\rangle}}

\newcommand{\ketbra}[2]{\mathinner{| #1\rangle \langle #2|}}
\newcommand{\dyad}[1]{| #1\rangle \langle #1|}  	

\renewcommand{\a}{\alpha}
\renewcommand{\b}{\beta}

\newcommand{\overbar}[1]{\mkern 1.5mu\overline{\mkern-1.5mu#1\mkern-1.5mu}\mkern 1.5mu} 
\newcommand{\code}[1]{\texttt{#1}}  

\begin{document}

\title  {Positive maps and trace polynomials from the symmetric group}
\date   {\today}
\author {Felix Huber}
\affiliation{ICFO - The Institute of Photonic Sciences, 08860 Castelldefels (Barcelona), Spain}

\maketitle

\noindent {\bf Abstract:} 
With techniques borrowed from quantum information theory,
we develop a method to systematically obtain operator inequalities and identities in several matrix variables.
These take the form of trace polynomials: polynomial-like expressions that involve matrix monomials 
$X_{\alpha_1} \cdots X_{\alpha_r}$ and their traces $\operatorname{tr}(X_{\alpha_1} \cdots X_{\alpha_r})$.
Our method rests on translating the action of the symmetric group on tensor product spaces into that of matrix multiplication.
As a result, we extend the polarized Cayley-Hamilton identity to an operator inequality on the positive cone,
characterize the set of multilinear equivariant positive maps 
in terms of Werner state witnesses, 
and construct permutation polynomials and tensor polynomial identities on tensor product spaces.
We give connections to concepts in quantum information theory and invariant theory.

\section{Introduction}

The study of polynomials that are positive or identically zero 
on certain sets has a rich history, 
going back to Hilbert's 17th problem.
More recently, problems in control theory and optimization led to the study of linear matrix inequalities, 
non-commutative Positiv- and Nullstellensätze, as well as the formulation of semidefinite relaxations 
for polynomials in non-commutative variables~\cite{
helton2006positive,
BlekhermanParriloThomas2012,
10.2307/3844994,
KLEP20081816,
PhysRevA.69.022308,
doi:10.1137/090760155,
BurgdorfKlepPovh2016}.
An important result is that by Helton, which states that 
all positive non-commutative polynomials are sums of hermitian squares~\cite{10.2307/3597203}. 
(An analogous result does not hold for commutative variables).

In this article, we focus on the larger class of {\em trace polynomials}:
expressions which can be realized as products and linear combinations of 
matrix monomials $X_{\alpha_1} \cdots X_{\alpha_r}$ and their traces $\tr(X_{\alpha_1} \cdots X_{\alpha_r})$.
We limit ourselves to {\em multilinear} expressions, 
that is, those which are linear in each variable.
To state an example, the expression $X_1 X_2X_3 + \tr(X_2)X_3X_1 - 2\tr(X_1X_3)\tr(X_2)$ is a multilinear trace polynomial. [We interpret the last term as $2\tr(X_1X_3)\tr(X_2)\one$.]
We study the following question, illustrated in Figure~\ref{fig:set}: 
which multilinear trace polynomials are positive on the set of positive semidefinite $d\times d$ matrices?

\begin{figure}[ht!]
 \includegraphics[width=0.52\linewidth]{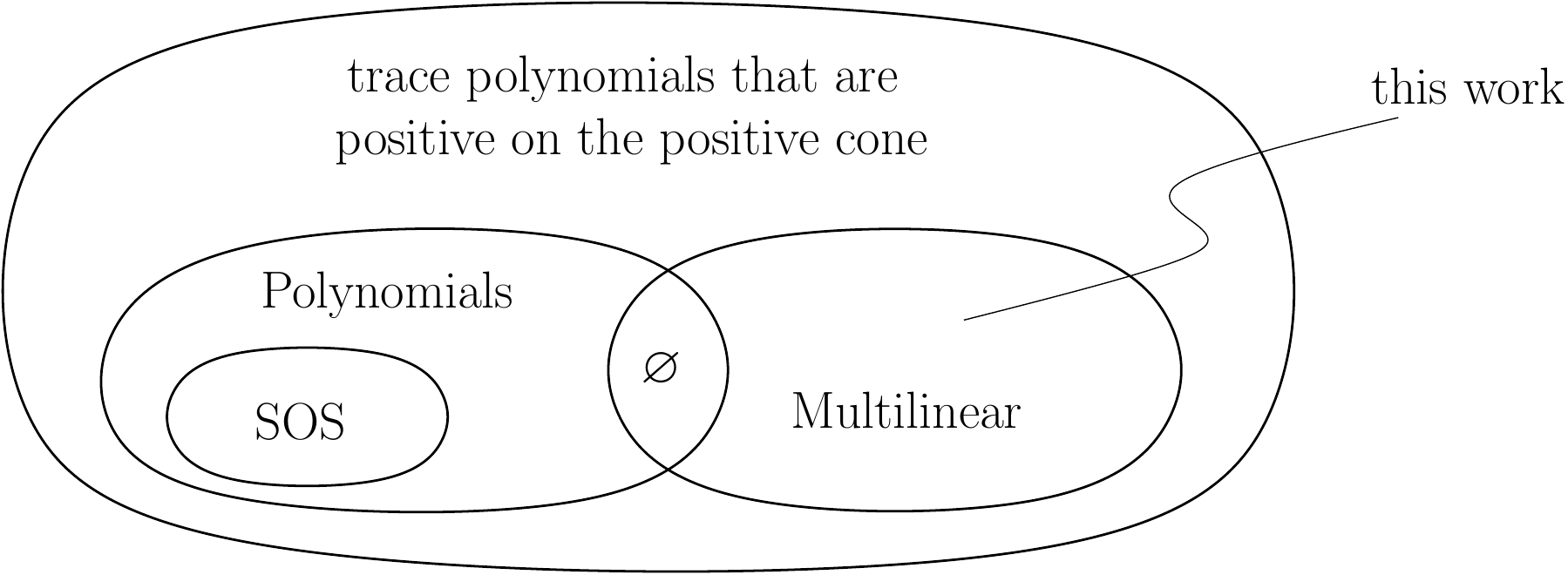}
 \caption{The set of trace polynomials that are positive on the positive cone (equivariant positive maps). 
 It contains expression which are polynomials, for example sum-of-squares (SOS).
 In this work, we characterize those trace polynomials that multilinear and positive on the positive cone.
 \label{fig:set}}
\end{figure}

This question is motivated from applications in quantum information theory.
Here the subset of positive semidefinite matrices, also known as the {\em positive cone},
takes a central role. 
Namely, quantum states are represented by complex positive semidefinite matrices of unit trace.
 As a consequence, one frequently requires inequalities and identities for this subset of matrices
in terms of the Löwner order: we write $A\geq B$ when $A-B$ is a positive semidefinite matrix,
and $A=0$ if $A$ is equal to the zero matrix.
 
Linear maps from the positive cone to itself are known as {\em positive maps}.
These have applications chiefly in the study of quantum dynamics and entanglement~\cite{
Christandl2019, 
PhysRevA.93.042335,
doi:10.1063/1.5045559,
PhysRevA.78.062105,
doi:10.1063/1.4927070,
PhysRevLett.113.100501,
doi:10.1063/1.4998433,
Johnston2019nonmpositive}.
Rather surprisingly, they are also useful in other contexts:
to obtain monogamy of entanglement relations,
to find bounds on the parameters of quantum codes, 
and 
to give constraints on the quantum marginal problem~\cite{
PhysRevA.64.042315, 
PhysRevA.72.022311, 
Butterley2006, 
PhysRevA.75.032102,
Eltschka2018distributionof,
PhysRevA.98.052317, 
651000,
796376,
817508,
681316, 
1751-8121-51-17-175301,
Huber2020quantumcodesof,
Huber_Thesis}. 
More general maps (not necessarily positive) of trace polynomial form
are of interest also in the context of 
joint measureability~\cite{DesignolleFarkasKaniewsk2019}.
\bigskip

Our second focus lies on {\em tensor polynomials}, that is, 
tensor products of non-commutative polynomials and linear combinations thereof. For example, 
$X_1 X_2 X_3 \ot \one + X_2 \ot X_3 X_1 - 2 X_1 X_3 \ot X_2 $
is a tensor polynomial ``living'' on two tensor factors. 
Here we ask:
how can one characterize {\em tensor polynomial identities}, i.e. 
tensor polynomials that vanish on the set of all $d \times d$ matrices? 
Furthermore, what expressions yield operators that permute the individual tensor factors (generalized swap operators)?

\begin{figure}[h!]
 \includegraphics[width=0.52\linewidth]{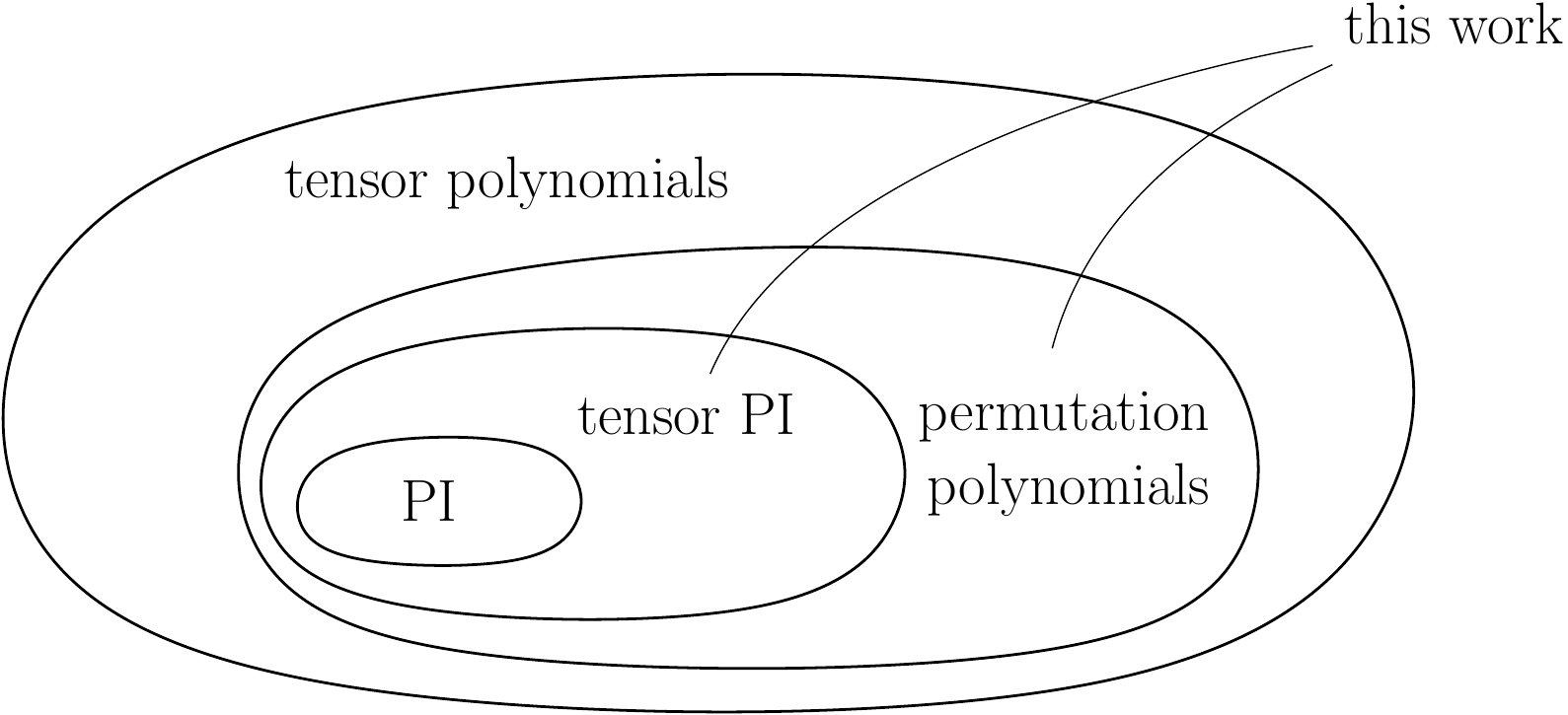}
 \caption{The set of tensor polynomials 
 contains expression that vanish identically zero on the set of $d \times d$ matrices, 
 such as polynomial identities (PI) and tensor polynomial identities (tensor PI).
 In this work, we characterize tensor polynomials that either 
 vanish or that yield permutation operators.
 \label{fig:set2}}
\end{figure}

Again, these questions are motivated from applications in physics.
Polynomial identities are used
as dimensional constraints in semidefinite programming hierarchies~\cite{
PhysRevA.92.042117}
and central polynomials 
as cut-and-glue operators and bond dimension witnesses for matrix product states~\cite{
Navascues2018bonddimension}. 
Multipartite quantum systems are represented by positive operators of unit trace 
on tensor product spaces. 
Thus tensor polynomial identities and permutation polynomials which
``mix'' the different tensor factors are naturally interesting,
for example in manipulating the time evolution of quantum systems~\cite{
Trillo2019,
PhysRevX.8.031008}. 
It is likely that one can also derive rank detection criteria from them.

\bigskip 

In mathematics, our questions connect to the invariant theory of matrices~\cite{
Formanek1989, 
DrenskyFormanek2004, 
ConciniProcesi2017}. 
Multivariate expressions, 
polynomial in the entries of the matrix variables 
and invariant under the simultaneous conjugate unitary action 
on all variables by unitaries,
are termed {\em polynomial invariants}. 
The set of polynomial invariants is generated by traced matrix monomials.
It can be shown that every polynomial invariant~$\iota$ 
is related to a multilinear equivariant map~$f$ by
\begin{equation}\label{eq:inv_eq}
 \iota(X_1, \dots, X_k, X_{k+1}) = \tr[f(X_1, \dots, X_k) X_{k+1}]\,.
\end{equation} 
Here, a map $f$ is termed {\em equivariant}, 
if $U f(X_1, \dots, X_k)U^\dag = f(U X_1 U^\dag, \dots, U X_k U^\dag)$
holds for all complex matrices $X_1,\dots, X_k$ and unitary matrices $U$.
From Eq.~\eqref{eq:inv_eq} it follows that the set of equivariant maps equals the set of trace polynomials.
It can also be seen that polynomial invariants that vanish on the set of $d \times d$ matrices ({\em trace identities}) 
can be turned into trace polynomials that evaluate to the zero matrix
({\em trace polynomial identities}).

\bigskip

Here we use a similar strategy to achieve the following:
first, we turn certain polynomial invariants,
known to be {\em positive} (i.e. non-negative) on the positive cone,
into equivariant positive maps. 
Second, we convert suitable trace identities 
into tensor polynomial identities
for the set of $d \times d$ matrices.
This provides us with a systematic method to work with matrix identities and inequalities.
First, it reduces the characterization of trace polynomials that are positive on the positive cone 
to that of invariant block-positive operators, e.g. entanglement witnesses for Werner states. 
Second, we show that tensor polynomial identities 
correspond to certain sufficiently 
anti-symmetric elements in the group algebra of the symmetric group, 
for which constructions can be given.
In the theory of polynomial identity rings, 
it was shown that all multilinear trace identities 
arise as a consequence of the Cayley-Hamilton theorem~\cite{Formanek1989}.
We show that the same result holds in the tensor setting.

We set our approach in contrast with recent works in non-commutative 
optimization and geometry:
Ref.~\cite{doi:10.1112/plms.12156} gives Positivstellensätze 
for the positivity of symmetric trace polynomials of $n \times n$ matrices 
on semialgebraic sets given by trace polynomial constraints.
Ref.~\cite{klep2020optimization} presents Positivstellensätze 
to obtain a hierarchy of semidefinite relaxations to approximate 
the minimum of dimension-free pure symmetric trace polynomials 
under trace polynomial constraints evaluated on von Neumann algebras
(with normalized traces).
In contrast to these works, we focus on the case of
not necessarily symmetric
trace polynomials with finite dimensional matrices as variables;
these are are either identities, 
evaluate to permutation operators, 
or are positive on positive semidefinite matrices.

\bigskip
The article is structured as follows:
we give an overview on the contributions, on related work from quantum information theory, 
and on the proof strategy in the remainder of this Section.
Section~\ref{sec:prelim} fixes notation and states preliminaries,
while Section~\ref{sec:matmultperm} introduces the key method of translating
permutations acting on tensor spaces into matrix multiplications.
Our main results are in Sections~\ref{sect:pCH}-\ref{sect:PI}: 
Section~\ref{sect:pCH} develops a family of positive equivariant maps 
that have the form of the Cayley-Hamilton identity.
In Section~\ref{sect:matrix_ineq}, we characterize the complete set 
of equivariant positive maps (or trace polynomial inequalities) in terms of invariant block-positive operators and entanglement witnesses.
Section~\ref{sect:PI} develops connections to the invariant theory of matrices.
We outline facts about trace identities and use these to give a complete characterization 
of the set of tensor polynomial identities and related concepts,
before giving conclusions in Section~\ref{sec:conclusion}.

Appendix~\ref{app:RepTheorySk} 
sketches the representation theory of the symmetric group and the construction of Young projectors, 
Appendix~\ref{app:tables} 
provides tables of equivariant positive maps and matrix inequalities, 
and 
Appendix~\ref{app:construction_polyId} 
give recipes for the construction tensor polynomial identities and permutation polynomials.

\subsection{Contributions}\label{sect:contributions}

\begin{enumerate}
\item 
    We construct a family of trace polynomials that are positive on the positive cone, 
    that is, a family of positive equivariant multilinear maps.
    This extends the {\em polarized Cayley-Hamilton identity} 
    to an operator inequality on the positive cone.
    The trace identity as proven by Lew in 1966 states that expressions such as 
    \begin{equation} \label{eq:example0}
    f(X_1,X_2) 
    = \tr(X_1)\tr(X_2){} - \tr(X_1 X_2){}  - \tr(X_1)X_2 - \tr(X_2)X_1 + X_1 X_2 +  X_2 X_1
    \end{equation}
    are identically zero when evaluated on $2\times 2$ 
    matrices~\cite{Lew1966}.
    We show that $f$ and its generalization~$f_\lambda$
    to arbitrary many matrix variables are positive maps.
    For our example above, this means that $f(X_1,X_2)\geq 0$ 
    whenever $X_1,X_2 \geq 0$.
    We show that this family of maps can
    be characterized as 
    completely copositive, 
    equivariant under unitaries, 
    and tensor-stable. 
    
    The details can be found in Section~\ref{sect:pCH}.
    Appendix~\ref{app:tables} contains tables for these maps in up to $3$ variables.

\item 
    We give a characterization of the set of multilinear equivariant positive maps.
    Every map of this type corresponds to an invariant block-positive operator; 
    whereas every {\em optimal} map corresponds to an optimal Werner state witness. 
    As a consequence, the set of symmetric multilinear equivariant positive
    maps has an infinite number of extreme points in the case of four or more variables.
     
    The details can be found in Section~\ref{sect:matrix_ineq}.
    
\item
    We connect our methods to those of invariant theory and 
    polynomial identity rings. 
    Here we give a complete characterization of polynomial 
    identities on tensor product spaces in terms of 
    certain sufficiently anti-symmetric elements 
    of the group algebra $\C[S_k]$. 
    We show that all multilinear permutation polynomials and tensor polynomial
    identities arise as consequences of the Cayley-Hamilton theorem.
     
    The details can be found in Section~\ref{sect:matrix_ineq}.
    Appendix~\ref{app:construction_polyId} gives a list of constructions.
\end{enumerate}

Overall, 
we extend the connection between identities for polynomial invariants 
and equivariant maps, a concept well-known in the study of 
polynomial identity rings~\cite{
Kostant2009, 
DrenskyFormanek2004, 
ConciniProcesi2017}, 
to that of inequalities and to the setting of tensor product spaces.

\subsection{Related concepts from quantum information theory}\label{sect:related}

Our approach generalizes and relates to concepts from in quantum information theory:

\begin{enumerate}[label=(\roman*)]

 \item {\em Werner-Holevo channel:}  
 the map $\rho \mapsto \one - \rho^T$ is a completely positive map~\cite{
 doi:10.1063/1.1498491, 
 HaydenWinter2008}. 

 \item {\em Werner states:} 
 these are multipartite quantum states which satisfy
 $\rho = U^{\ot n} \rho (U^\dag)^{\ot n}$ for all unitaries $U \in \UU(d)$.
 Witnesses that detect entanglement in Werner states show the same invariance, 
 $\WW = U^{\ot n} \WW (U^\dag)^{\ot n}$.

 \item {\em Reduction criterion:} 
 let $\rho_{AB}$ be a bipartite quantum state.
 If $\tr_{B}(\rho_{AB}) \otimes \one_{B}- \rho_{AB} \not \geq 0$, then $\rho_{AB}$ is entangled
 and a maximally entangled state can be distilled from finitely many copies of $\rho_{AB}$~\cite{
 PhysRevA.60.898, 
 PhysRevA.59.4206}.
 
 \item{\em Universal state inversion:}
  the map $\rho \mapsto \sum_{S \subseteq \{1,\dots, n\}} (-1)^{|S|} \tr_{S^c}(\rho) \ot \one_{S^c}$
  is positive but not completely~\cite{
  PhysRevA.64.042315, 
  PhysRevA.72.022311}.

 \item {\em Few-copy entanglement detection, randomized measurements, 
 and concurrences:}
 measureing the overlap of projectors onto invariant subspaces of the symmetric group 
 with that of multiple copies of a quantum state, i.e. expressions of the form $\tr(P_\lambda \rho^{\otimes n})$, results in invariant entanglement detection criteria closely related to generalized concurrences~\cite{
 PhysRevLett.98.140505, 
 PhysRevLett.95.260502,
 PhysRevLett.104.210501}. 
 These can be performed with few-copy or randomized measurements~\cite{
 PhysRevLett.108.110503,
 PhysRevA.99.052323}.
 
 \item {\em Marginal compatibility:} 
 for a tripartite state $\rho_{ABC}$ with the marginals $\rho_{AB}, \rho_{AC}, \rho_{BC}$ to exist, 
 it is necessary that~\cite{Butterley2006}
 \begin{equation}\label{eq:butterley}
  \one - \rho_A - \rho_B - \rho_C + \rho_{AB} + \rho_{AC} + \rho_{BC} \geq 0\,.
 \end{equation}
 where the marginals are understood to be tensored 
 with the identity matrix on the remaining subsystems.
 Similar criteria appear in Refs.~\cite{
 PhysRevA.75.032102, 
 PhysRevA.98.052317} 
 and can be interpreted as quantum Fréchet inequalities~\cite{PhysRevA.94.042106}.

\item {\em Bounds for quantum codes and the shadow inequality:} 
    Rains' shadow inequality states that for all $M,N \geq 0$ and all $|T|\subseteq \{1,\dots, n\}$ that
    \begin{equation}\label{eq:shadow_intro}
     \sum_{ S \subseteq \{1,\dots, n\}} (-1)^{|S\cap T|} \tr[\tr_{S^c}(M)\tr_{S^c}(N)] \geq 0\,.
    \end{equation} 
    Including this inequality into the 
    linear programming bounds on the 
    weight enumerators of quantum codes 
    gives stronger bounds on the existence of codes~\cite{
    681316, 
    796376,
    817508}.
    For qubit codes, an Ansatz for the weight enumerator is known and 
    Rains has shown analytically that 
    the distance $d$ of any $n$-qubit quantum error-correcting code is bounded by
    $d \lesssim \frac{n}{3}$~\cite{651000}.
    These coding bounds can be understood as monogamy of entanglement constraints 
    [c.f. Eq.~\ref{eq:monogamy}] applied the projector onto the code space. 
    For absolutely maximally entangled states and quantum MDS codes the constraints can be evaluated directly, often yielding tight bounds for small systems~\cite{
    1751-8121-51-17-175301, 
    Huber2020quantumcodesof,
    Huber_Thesis}.
    We note that Eq.~\eqref{eq:shadow_intro} can be turned into the (shadow) operator inequality~\cite{
    Eltschka2018distributionof, 
    PhysRevA.98.052317}
    \begin{equation}
        \sum_{ S \subseteq \{1,\dots, n\}} (-1)^{|S\cap T|} \rho_S \ot \one_{S^c} \geq 0\,,
    \end{equation}
    whose generalization forms the motivation of this work.
    
 \item {\em Monogamy of entanglement:} 
 let $\ket{\psi}$ be a multipartite finite-dimensional state with reductions 
 $\rho_S = \tr_{S^c}(\dyad{\psi})$,
 where $S^c$ is the complement of $S$ in $\{1,\dots, n\}$.
 For all $T\subseteq \{1,\dots, n\}$ 
 the distribution of concurrences 
 $C_{S|S^c}(\psi) = \sqrt{2[1- \tr (\rho_S^2)}]$
 is constrained by~\cite{
 Eltschka2018distributionof, 
 PhysRevA.98.052317}
 \begin{equation}\label{eq:monogamy}
  - \sum_{S \subseteq \{1,\dots,n\}} (-1)^{|T \cap S|} C_{S|S^c}^2(\psi) \geq 0\,.
 \end{equation}

\item {\em Finding joint measurements:} equivariant maps can be used as Ansätze for joint measurements.
    For two POVMs $\{A_a\}$ and $\{B_b\}$, joint measurements of the form
    \begin{equation}
     G_{ab} = \a\one\tr(A_aB_b) + \b \one\tr(A_a)\tr(B_b) + \gamma [\tr(A_a)B_b + \tr(B_b)A_a] + \delta (A_a B_b+B_b A_a)
    \end{equation}
    with $\a,\b,\gamma, \delta \in \R$ can be found efficiently 
    with a semidefinite program~\cite{DesignolleFarkasKaniewsk2019}.
    
\item {\em Time manipulation:} 
      tensor polynomials for which $\sum_i p_i \ot q_i \propto \operatorname{SWAP}$ 
      (where $p_i, q_i$ are non-commutative polynomials)
      for all $d\times d$ matrices can be 
      used to manipulate the time evolution of quantum systems~\cite{Trillo2019}.

\item {\em Bond dimension witnesses and cut-and-glue operators for matrix product states:}
        polynomial identities that vanish on all $d \times d$ matrices can be used to 
        witness the bond dimension in matrix product states.
        Central polynomials evaluating to the identity can be used to decouple particles 
        from a spin chain in a cut-and-glue fashion~\cite{Navascues2018bonddimension}.
        
\end{enumerate}

\subsection{Proof strategy}\label{sect:proof_strat}

The proofs rely on the following chain of reasoning:
denote by~$T$ the representation of the symmetric group that exchanges the
tensor factors of $(\C^d)^{\ot k}$. Let $(k\dots 1)$ be the inverse of the permutation $(1\dots k)$
and let $X_1, \dots, X_k$ be $d \times d$ matrices. Then
\begin{equation}\label{eq:sketch1}
    \tr_{1 \dots k \backslash k}  \big[ T{((k \dots 1))} \, X_1 \ot X_2 \ot \cdots \ot X_k\big]
    = X_1 X_2 \cdots X_k\,.
\end{equation}
Note that the trace is taken over all tensor factors except the last one.
Thus, the action of a permutation on a tensor product space can be translated
into that of a matrix multiplication.

Consider now some positive semidefinite operator $\PP$ acting on $(\C^d)^{\ot k}$ 
and replace the last variable $X_k$ in Eq.~\eqref{eq:sketch1} by the identity matrix,
\begin{equation}\label{eq:sketch2}
    \tr_{1 \dots k \backslash k}  \big[ \PP \, X_1 \ot X_2 \ot \cdots \ot X_{k-1} \ot \one \big] \,.
\end{equation}
We recall that an operator $A$ is positive semidefinite if and only if $\tr[A B] \geq 0$ for all $B\geq 0$ holds; 
this is known as the self-duality of the positive cone.
We check this for the expression in Eq.~\eqref{eq:sketch2},
\begin{equation}\label{eq:sketch3}
    \tr\Big\{ \tr_{1 \dots k \backslash k}  \big[ \PP \, X_1 \ot X_2 \ot \cdots \ot X_{k-1} \ot \one \big] \cdot B\Big\}
    = \tr \big[ \PP\, X_1 \ot X_2 \ot \cdots \ot X_{k-1} \ot B \big] \,,
\end{equation}
where we made use of the coordinate-free definition of the partial trace: 
$\tr\big[ \tr_1(M) N\big] = \tr\big[M (\one \ot N)\big]$ holds for all operators $M$ and $N$ 
acting on Hilbert spaces $\HH_1 \ot \HH_2$ and $\HH_2$ respectively.
When $X_1, \dots, X_{k-1},B \geq 0$, then Eq.~\eqref{eq:sketch3} is non-negative 
because it is the Hilbert-Schmidt inner product of the two positive semidefinite operators
$\PP$ and $X_1 \ot \dots \ot X_{k-1} \ot B$.
Consequentially, the expression in Eq.~\eqref{eq:sketch2} is positive semidefinite 
and can be understood as a multilinear positive map with the $X_i$ as variables.

Naturally, we want to make use of the relation in Eq.~\eqref{eq:sketch1}.
We thus choose some $\PP\geq 0$ that is a linear combination of permutations.
Then Eq.~\eqref{eq:sketch2} yields trace polynomials that are positive on the positive cone, 
in other words, multilinear equivariant positive maps.
For example, the idempotent 
\begin{equation}
 \omega= \frac{1}{6}[()-(12)-(13)-(23)+(123)+(132)] \,\in\, \C S_3
\end{equation} 
yields the projector onto the completely anti-symmetric subspace 
$P_{\omega} = \hat T(\omega) \geq 0$. 
We recover Eq.~\eqref{eq:example0},
 \begin{equation}\label{eq:example00}
  \tr_{12}[P_{\omega}\, X \ot Y \ot \one] = \tr(X)\tr(Y){} - \tr(XY){} - \tr(X)Y - \tr(Y)X + XY + YX \geq 0 
  \quad \text{whenever}\quad X,Y\geq 0\,.
 \end{equation}

This motif can be explored further:
first, by choosing the support of $\PP$ to rest exclusively 
in representations that are ``too anti-symmetric'' for a vector space of dimension $d$. 
It will turn out that this idea completely characterizes 
permutation polynomials and tensor polynomial identities. 
Second, a careful look at Eq.~\eqref{eq:sketch3} shows 
that it is not required that $\PP$ be positive semidefinite
for Eq.~\eqref{eq:sketch3} to be non-negative.
A merely blockpositive operator, 
that has non-negative expectation values on the set of separable states, 
suffices.
Thus one can obtain a correspondence between the set of multilinear equivariant positive maps 
and the set of unitary invariant blockpositive operators.
Their boundary of either set bijects to the set of optimal Werner state witnesses,
yielding a complete characterization.

\section{Notation and preliminaries}\label{sec:prelim}

\subsection{The cone of positive matrices, quantum states, and positive maps}
Denote the set of complex $d \times d$ matrices as $\MM_d$ and
the cone of positive-semidefinite $d \times d$ matrices as~$\MM_d^{+}$, 
also known as the {\em positive cone}.
Its natural (semi-) order is that of Löwner, where $A\geq B$ if $A-B$ is positive semidefinite.
We write $\id_d$ or $\one_d$ for the identity map on $\MM_d$ and $A=0$ if $A$ is the zero matrix.
The matrix transpose is $\theta(X) = X^T$, 
and we write the partial transpose on a subsystem $S$ as $\theta_S(X) = X^{T_S}$.
We denote the space of linear operators acting on a vector space $V$ as~$L(V)$. 
Lastly, the set of unitary $d\times d$ matrices is~$\UU(d)$.

The trace of a matrix $X$ is the sum of its diagonal elements, $\tr(X) = \sum_i X_{ii}$. Given two complex square matrices $X$ and $Y$ of the same size, their Hilbert-Schmidt inner product 
is defined as $\langle X, Y \rangle = \tr(X^\dag Y)$. 
It is known that $\tr(X^\dag Y) \geq 0$ when $X$ and $Y$ are positive semidefinite.
In this context we recall a useful characterization of hermitian and positive semidefinite matrices~\cite[Theorem 4.1.4]{horn_johnson_2012}:
a $d \times d$ matrix $X$ is hermitian if and only if $\bra{\phi} X \ket{\phi} \in \R$ for all $\ket{\phi} \in \C^d$ holds;
$X$ is positive semidefinite if and only if $\bra{\phi} X \ket{\phi} \geq 0$ for all $\ket{\phi} \in \C^d$ holds.
It follows that the expression $\tr(X^\dag Y)$ is real and nonnegative for all $Y \in \MM_d^{+}$ if and only if 
also $X\in \MM_d^{+}$. 
This fact is known as the {\em self-duality} of the positive cone and
$\MM_d^{+}$ equals its dual cone $\{ X \,|\, \tr(X^\dag Y) \geq 0\,, \forall \,Y \in \MM_d^{+}\}$.

The set of quantum states with $d$ levels is formed by the set of hermitian positive semidefinite matrices $\rho$ of trace one,
also known as {\em density matrices}. 
That is, a quantum state satisfies $\rho \geq 0$, $\rho^\dag = \rho$, and $\tr \rho = 1$.
Quantum states composed of $n$ particles with $d$ levels each are then elements in $(\MM_d \ot \dots \ot \MM_d)^{+} \, (n \text{ times})$.
Multipartite states are called {\em entangled} if they cannot be written as a convex combination of product states, 
$\rho^\text{ent} \neq \sum_i p_i \rho_1^{(i)} \ot \dots \ot \rho_n^{(i)}$ with $\sum_i p_i = 1$ and $p_i \geq 0$.
States that are not entangled are {\em separable} and are elements of 
$\MM_d^{+} \ot \dots \ot \MM_d^{+} \, (n \text{ times})$.

A linear map $\Lambda : M_{d_1} \to M_{d_2}$ is called {\em positive}, 
if $\Lambda(X) \geq 0$ whenever $X \geq 0$.
A linear map is termed {\em completely positive} (CP), if 
$\Lambda_{d} \ot \id_m$ is positive for all $m \in \N_+$.
A map $\Lambda$ is called {\em completely copositive} (coCP), 
if $\Lambda \circ \theta$ is completely positive.
A positive map $\Lambda$ is termed {\em tensor-stable}, if $\Lambda^{\otimes n}$ is positive for all $n \in \N$~\cite{Christandl2019}.
Naturally, all completely positive maps are tensor-stable, 
while maps that are completely copositive remain so under tensor powers.

A map is {\em multilinear} if it is linear in each variable. 
We call a multilinear map 
$\Lambda : M_{d}^n \to M_{d'}$
{\em positive}, if $\Lambda(X_1, \dots, X_n) \geq 0$  whenever $X_1, \dots, X_n \geq 0$.
We also say that these maps are positive on the positive cone.
A map $f$ is termed {\em equivariant}, 
if $U f(X_1, \dots, X_k)U^\dag = f(U X_1 U^\dag, \dots, U X_k U^\dag)$
holds for all complex matrices $X_1,\dots, X_k$ and unitary matrices $U$.

\subsection{Non-commutative and trace polynomials}
Naturally, matrices in general do not commute.
Consider some collection of matrix variables $\{X_1, \dots, X_n\}$.
The set of {\em non-commutative polynomials} 
is formed by linear combinations of monomials $X_{\a_1} \cdots X_{\a_r}$,
commonly denoted by~$\R\langle X \rangle$.
The algebra of {\em trace polynomials} is generated by monomials 
$X_{\a_1} \cdots X_{\a_r}$ over the ring of traces $\tr(X_{\a_1} \cdots X_{\a_r})$.
Here any trace polynomial term that only contains traced expressions,
e.g. $\tr(X_{\a_1} \cdots X_{\a_r}) \cdots \tr(X_{\zeta_1} \cdots X_{\zeta_t})$, 
is interpreted as the scalar matrix, 
$\tr(X_{\a_1} \cdots X_{\a_r}) \cdots \tr(X_{\zeta_1} \cdots X_{\zeta_t}) \one$. 
Similarly, {\em tensor (trace) polynomials} are given by linear combinations of tensor products of
non-commutative (trace) polynomials.

\begin{table}
\begin{tabular}{@{} l @{\quad\quad} l@{} }
\toprule
non-commutative polynomial  & $X_1 X_2 X_3 + X_2 X_3 X_1 -  2 X_1 X_3 X_2 $                      \\
trace polynomial            & $X_1 X_2 X_3 + \tr(X_2) X_3 X_1 - 2 \tr(X_1X_3)\tr(X_2)$       \\
polynomial invariant
                            & $\tr(X_1 X_2 X_3) + \tr(X_2)\tr(X_3 X_1) - 2 \tr(X_1X_3)\tr(X_2)$  \\
\midrule
tensor polynomial           & $X_1 X_2 X_3 \ot \one +  X_2 \ot X_3 X_1 - 2 X_1 X_3 \ot X_2 $           \\
tensor trace polynomial     & $ \tr(X_1 X_2) X_3 \ot \one + X_2 \ot X_3 X_1 - 2 \tr(X_1) X_3 \ot X_2 $ \\
\bottomrule
\end{tabular}
\caption{\label{table:TP_related} Non-commutative polynomials and related objects.
Terms containing only traces are interpreted as a scalar matrix; i.e. read $\tr(X_1 X_3)$ as $\tr(X_1X_3)\one$.}
\end{table}

A {\em polynomial invariant}
is a scalar expression that is polynomial in the entries of 
the matrix variables and invariant under the simultaneous conjugate action 
on all variables by unitaries.
It has been shown that the ring of polynomial invariants is 
generated by traced matrix monomials ({\em pure} trace polynomials).
Every polynomial invariant~$\iota$ 
is related to a multilinear equivariant map~$f$ by
\begin{equation}\label{eq:inv_eq2}
 \iota(X_1, \dots, X_k, X_{k+1}) = \tr[f(X_1, \dots, X_k) X_{k+1}]\,.
\end{equation} 
From Eq.~\eqref{eq:inv_eq2} it follows that the set of multilinear equivariant maps and 
the set of multilinear trace polynomials coincide~\cite{ConciniProcesi2017} (see also Section~\ref{sect:PI}).
Examples for these different types of generalized polynomials is given 
in Table~\ref{table:TP_related}.

A {\em polynomial identity} is a non-commutative polynomial that vanishes 
on the set of $d \times d$ matrices $\MM_d$ for some $d$.
Likewise, a {\em trace identity} is a trace polynomial that vanishes,
and tensor (trace) polynomial identities are tensor (trace) polynomials 
that vanish on $\MM_d$.
As in the case of positive maps, we call a non-commutative polynomial {\em positive} 
if it is positive semidefinite whenever all variables are positive semidefinite.  
We also say that the polynomial is ``positive on the positive cone''.
Note that this nomenclature is different than that of e.g. Ref.~\cite{10.2307/3597203},
where positive semidefiniteness is required on the set of {\em all} $d\times d$ matrices.

\subsection{Partial trace, Choi-Jamiołkowski isomorphism, and Swap}
The {\em partial trace} is used in quantum mechanics to obtain the 
reduced or local description of quantum states. 
For example, given a bipartite quantum state $\rho_{12}$ and some orthonormal basis $\{\ket{i}_2\}_{i=1}^d$ for the second subsystem, 
the partial trace is commonly written as 
$\tr_2(\rho_{12}) = \sum_{i=1}^d \bra{i}_2 \rho_{12} \ket{i}_2$.
The state $\rho_1 = \tr_2(\rho)$ then gives the complete information 
on measurement outcomes when considering system $1$ alone.
Here, we will here need the more abstract coordinate-free definition of the partial trace:
denote by $\HH_1$ and $\HH_2$ two Hilbert spaces. 
Then, the partial trace $\tr_1$ is the unique linear operator for which
\begin{equation}\label{eq:coord_free_def_ptrace}
 \tr\big[ \tr_1(M) N\big] = \tr\big[M (\one \ot N)\big] 
\end{equation}
holds for all operators $M$ and $N$ 
acting on Hilbert spaces $\HH_1 \ot \HH_2$ and $\HH_2$ respectively.
In other words, the partial trace is 
the adjoint operation to $M \to M \ot \one$
for the Hilbert-Schmidt inner product $\langle A, B \rangle = \tr(A^\dag B)$.
 
The {\em Choi-Jamiołkowski isomorphism} relates 
the space of linear maps $L(V)$
with $V \ot V$.
Let $\Lambda \in L(V)$ be a linear map. 
Define the correspondence
\begin{equation}
 \Lambda \rightarrow \rho_\Lambda  = (\Lambda \ot \one) \dyad{\Omega}
\end{equation}
where $\ket{\Omega} = \sum_{i=1}^d \ket{ii}$ is the (unnormalized) maximally entangled state.
The inverse is given by
\begin{equation}\label{eq:CJ_inverse}
 \rho \rightarrow \Lambda_\rho(X) = \tr_1[\rho (X^T \ot \one)]\,. 
\end{equation} 
The Choi-Jamiołkowski isomorphism states that $\Lambda_\rho$ 
is completely positive if and only if 
$\rho_\Lambda \in (\MM_d \ot \MM_d)^+$.

The {\em swap operator} $\Gamma$ exchanges the two tensor-components of a biproduct-vector,
that is \(\Gamma \ket{\phi} \ot \ket{\psi}= \ket{\psi} \ot \ket{\phi}\) for all $\ket{\phi},\ket{\psi} \in \C^d$.
It can be expanded as
\begin{equation}
 \Gamma
 = \dyad{\Omega}^{T_2} 
 = \sum_{j,k = 1}^d \ketbra{jj}{kk}^{T_2} 
 = \sum_{j,k=1}^d \ketbra{jk}{kj}\,,
\end{equation} 
where $(\cdot)^{T_2}$ denotes the partial transpose on the second subsystem.
It is not hard to establish (e.g. by direct matrix multiplication) that for all operators 
$M$ and $N$ acting on $\C^d$ the following relation holds,
\begin{align}\label{eq:swap_property}
 \tr [ \Gamma \,M \ot N] = \tr[ M N]\,. 
\end{align}

Under a {\em partial} trace, the swap results in the matrix multiplication of tensor-factors.
The resulting {\em swap identities} are well-known~(see e.g. Ref.~\cite{Eltschka2018distributionof}) 
and form the starting point for our work:
let $M$ and $N$ be $d\times d$ matrices. Then
\begin{align}\label{eq:swap_under_ptrace}
 \tr_1[\Gamma \,M \ot N] = MN \quad \text{and} \quad  \tr_2[\Gamma \,M \ot N] = NM\,.
\end{align}

\subsection{Action of the symmetric group on $(\mathds{C}^d)^{\otimes k}$}
\label{sect:act_sym_group}
Let \(S_k\) be the symmetric group, that is, the group of that permutes $k$ elements. 
Under the cycle notation, the permutation $(143)(2)$ maps 
$1\rightarrow 4 \rightarrow 3 \rightarrow 1$ and $2 \rightarrow 2$. 
Because $(143)(2)$ leaves the position $2$ invariant, 
we can further shorten the notation to $(143)$. 
The non-permutation is denoted by $()$ (sometimes also $e$) and is identical to $(1)(2)\dots(k)$.
We refer to $\pi(i)$ as the coordinate 
to which an object at coordinate $i$ is permuted to. 
Consequently, $\pi^{-1}(i)$ refers to what object was brought to position $i$ by~$\pi$. 
Thus if $\pi = (143)(2)$, then $\pi(4) = 3$ and $\pi^{-1}(4)=1$.
Given some permutation $\pi$, its {\em cycle structure} is given by the lengths and multiplicities of its cycles. 
Elements in $S_k$ are {\em conjugate} ($\pi_1 = \pi^{-1} \pi_2 \pi$ for some $\pi \in S_k$) 
if and only if they have the same cycle structure. 
Every permutation is conjugate to its inverse~\footnote{In terms of character theory, 
this is the case if and only if every character is real-valued.}.

Consider now the following representation $T$ of $S_k$ on a complex tensor-product space:
let \(T(\pi)\) act on $(\C^D)^{\ot k}$ by the permutation of its $k$ tensor factors according to $\pi$,
\begin{equation}\label{eq:natural_rep_of_Sk}
 T(\pi) \, \ket{v_1} \ot \dots \ot \ket{v_k} = \ket{v_{\pi^{-1}(1)}} \ot \dots \ot \ket{v_{\pi^{-1}(k)}}\,.
\end{equation}
For example, the permutation $\tilde \pi = (143)(2)$ acts on $(\C^D)^{\otimes 4}$ as
\begin{equation}
 T{(\tilde \pi)} \ket{v_1} \ot \ket{v_2} \ot \ket{v_3} \ot \ket{v_4} = \ket{v_3} \ot \ket{v_2} \ot \ket{v_4} \ot \ket{v_1} \,.
\end{equation}
The adjoint of $T$ acts in a reversed fashion on kets,
$T(\pi)^\dag  \ket{v_1} \ot \dots \ot \ket{v_k} = \ket {v_{\pi(1)}} \ot \dots \ot \ket{v_{\pi(k)}}$.
One can check that the representation $T$ is unitary, $T^{\dag}(\pi) = T^{-1}(\pi) = T(\pi^{-1})$ for all $\pi \in S_k$.
The swap operator $T((ij))$ permutates two tensor factors $i$ and $j$.
When only two tensor factors are present we omit the indices altogether and write $\Gamma$.

A partition $\lambda$ of an integer~$k$ (written as $\lambda \vdash k$) is a sequence of 
positive integers $\lambda = (\lambda_1, \dots, \lambda_r)$, 
such that $\lambda_1 \geq \lambda_2 \geq \dots \geq \lambda_r$ and $\lambda_1 + \dots + \lambda_r = k$,
and $r$ is the number of parts.
We recall the Schur-Weyl duality.
\begin{theorem}[Schur-Weyl Duality~\cite{Audenaert2006}]
\label{thm:schur_weyl}
The tensor product space $(\C^d)^{\ot k}$ can be decomposed as
\begin{equation}
 (\C^d)^{\otimes k} = \bigoplus_{\substack{\lambda \vdash k \\ \text{parts}(\lambda) \leq d}} 
 \UU_\lambda \ot \SS_\lambda \,,
 \end{equation}
where the symmetric group $S_k$ acts on the spaces $\SS_\lambda$ and the general linear group $GL_d(\C)$ acts on the spaces $\UU_\lambda$, 
indexed by the same partitions.
\end{theorem}
Note that $(\C^d)^{\otimes k}$ does not contain subspaces that correspond to partitions with more than $d$ rows 
(c.f. Proposition~\ref{prop:pigeon} in Appendix~\ref{app:RepTheorySk}). 
This will be the origin of trace and polynomial identities, a topic which we will discuss in Section~\ref{sect:PI}.

In other words, the Schur-Weyl duality states that the diagonal action of the general linear group $GL_d(\C)$ 
of invertible complex $d \times d$ matrices and that of the symmetric group on $(\C^d)^{\ot n}$ 
commute. For all $A \in GL_d(\C)$ and $\pi \in S_n$,
\begin{equation}
 T(\pi) (A \ot \cdots \ot A) = (A \ot \cdots \ot A) T(\pi)\,.
\end{equation}

The projectors associated to the subspaces $\UU_\lambda \ot \SS_\lambda$ are the (central) {\em Young Projectors}~\footnote{
The usual Young projectors corresponding the Young symmetrizer are usually neither hermitian nor central.},
\begin{equation}
 P_\lambda = \frac{\chi_\lambda(e)}{k!} \sum_{\pi \in S_k} \chi_\lambda(\pi^{-1}) T(\pi)\,,
\end{equation} 
where $\chi_\lambda$ is the character associated to the irreducible representation indexed by $\lambda$,
and $e$ is the identity permutation in $S_k$.

For our purposes it is important that $P_\lambda = P_\lambda^\dag \geq 0$, 
that the central Young projectors commute with the action of the symmetric group 
and with the diagonal action of the general linear group,
and that they can be written as a linear combination of generalized swap operators $T(\pi)$.
We denote the group algebra representation corresponding to the representation $T$ by $\hat T$.
Then the Young projectors can equivalently be obtained from centrally primitive hermitian idempotents 
$\omega_\lambda$ in the group ring~$\C S_k$ as $P_\lambda = \hat T(\omega_\lambda)$. 
Further details about their construction and the representation theory of the symmetric group 
can be found in Appendix~\ref{app:RepTheorySk}.

\section{Matrix multiplication and permutations}\label{sec:matmultperm}

\subsection{Matrix products from permutations}
Our method to work with trace polynomials 
rests on generalizing the swap identities from the previous section [Eq.~\eqref{eq:swap_under_ptrace}].
We formalize the translation of permutations into matrix products.
  \begin{proposition}\label{prop:generalized_swap_identities}
    Let $X_1, \dots, X_k \in \MM_d$ with $k\geq 3$.
    Consider the cycle $(k \dots 1) = (1 \dots k)^{-1}$.
    Then
    \begin{align}
    \tr_{1 \dots k \backslash k}  \big[ T{((k \dots 1))} X_1 \ot X_2 \ot \cdots \ot X_k\big] 
    &= X_1 X_2 \cdots X_k \nn\\
    \tr_{1 \dots k \backslash 1}  \big[ T{((1 \dots k))} X_1 \ot X_2 \ot \cdots \ot X_k\big] 
    &= X_k X_{k-1} \cdots X_1\,.
    \end{align}
\end{proposition}

\begin{proof}
 Let $\{\ket{\a}\}$ be an orthonormal basis for $\C^d$. 
 Decompose 
   $X_{i} = \sum_{\a_i,\b_i =1}^d \chi_{\a_i \b_i}^{(i)} \ketbra{{\a_i}}{{\b_i}}$.
Then
\begin{align}
&\quad  \tr_{1 \dots k \backslash k}  \big[ T{((k\dots 1))} X_1 \ot X_2 \ot \cdots \ot X_k
\big] \nn\\
&= \tr_{1 \dots k \backslash k}  \big[ T{((k \dots 1))}
\sum_{\a_1,\b_1 =1}^d \chi_{\a_1 \b_1}^{(1)} \ketbra{{\a_1}}{{\b_1}} \ot
\sum_{\a_2,\b_2 =1}^d \chi_{\a_2 \b_2}^{(2)} \ketbra{{\a_2}}{{\b_2}} \ot
\dots \ot
\sum_{\a_k,\b_k =1}^d \chi_{\a_k \b_k}^{(k)} \ketbra{{\a_k}}{{\b_k}}
\big]\nn\allowdisplaybreaks\\
&= 
 \tr_{1 \dots k \backslash k}  \big[ \phantom{T{((k \dots 1))}}
\sum_{\a_1, \b_1 =1}^d \chi_{\a_1 \b_1}^{(1)} \ketbra{{\a_2}}{{\b_1}} \ot
\sum_{\a_2, \b_2 =1}^d \chi_{\a_2 \b_2}^{(2)} \ketbra{{\a_3}}{{\b_2}} \ot
\dots \ot
\sum_{\a_k, \b_k =1}^d \chi_{\a_k \b_k}^{(k)} \ketbra{{\a_{1}}}{{\b_k}}
\big] \nn \allowdisplaybreaks\\
&= 
\sum_{\a_1, \a_2, \dots, \a_k, \b_k =1}^d 
\chi_{\a_1 \a_2}^{(1)} \,
\chi_{\a_2 \a_3}^{(2)}
\dots 
\chi_{\a_{k-1} \a_k}^{(k)} \ketbra{{\a_{1}}}{{\b_k}}\nn\\
&= X_1 X_2 \cdots X_k\,.
\end{align}
The second relation can be shown in a similar way. This ends the proof.
\end{proof}

\subsection{Trace polynomials from permutations}

We now turn permutations into trace polynomials. For this it is helpful to introduce some additional notation.
Let a permutation $\pi \in S_k$ be given, decomposed into cycles as $\pi = \sigma_1 \dots \sigma_l$.
Naturally, cycles are equivalent whenever they differ by a cyclic shift of their elements. 
Here we demand a canonical ordering with which the elements are to be listed:
we require that the largest element of each cycle appear at its end, 
with the sequence of largest elements of cycles increasing.
In particular, $k$ appears then as the last item in the last cycle, $\sigma_l = (\dots k)$.
For example, the permutation written as $(3)(45)(216)$ is canonically ordered.

Let a set of matrices $X_1, \dots X_k$ be given. 
For a cycle $\sigma = (\sigma^{(1)} \sigma^{(2)} \dots \sigma^{(m)})$ 
we denote by $R_{\sigma}$ the product of $X_i$'s according to the 
canonical ordering with which the positions appear in the cycle.
\begin{equation}\label{eq:R_helper_function}
 R_{\sigma} = X_{\sigma^{(1)}} X_{\sigma^{(2)}} \cdots X_{\sigma^{(m)}}\,.
\end{equation}
For example, $R_{(314)} = X_3 X_1 X_4$.
For cycles of length one such as $\sigma = (i)$ one simply has $R_{\sigma} = R_{(i)} = X_{i}$.

We generalize Proposition~\ref{prop:generalized_swap_identities} 
to arbitrary permutations.
\begin{corollary}\label{cor:perm_and_ptrace}
    Let $X_1, \dots, X_k \in \MM_d$.
    Given a permutation $\pi \in S_k$ 
    consider its canonically ordered cycle decomposition $\pi = \sigma_1 \dots \sigma_l$.
    Then
    \begin{align}
    \tr_{1 \dots k \backslash k} \, [ T(\pi^{-1}) X_1 \ot X_2 \ot \cdots \ot X_k] 
    &= \tr(R_{\sigma_1})          \cdots \tr(R_{\sigma_{l-1}})   R_{\sigma_l} \,, \nonumber \\
    \tr_{ 1 \dots k \backslash \pi(k)} \, [ T(\pi) X_1 \ot X_2 \ot \cdots \ot X_k] 
    &= \tr(R_{\sigma_1^{-1}}) \cdots \tr(R_{\sigma_{l-1}^{-1}})   R_{\sigma_l^{-1}} \,.
    \end{align}
\end{corollary}
\begin{proof}
    The expression factorizes along the (disjoint) cycles and we make use of Proposition~\ref{prop:generalized_swap_identities} for the last term.
    \begin{align}
            \tr_{1 \dots k \backslash k}  [ T(\pi^{-1}) X_1 \ot X_2 \ot \cdots \ot X_k]
            &=  \tr [X_{\sigma_1^{(1)}} \cdots X_{\sigma_1^{(m_1)}}] 
                \cdots
                \tr [X_{\sigma_{l-1}^{(1)}} \cdots X_{\sigma_{l-1}^{(m_{l-1})}}]
                \cdot \tr_{\sigma_l \backslash k} [T(\sigma_l^{-1}) X_{\sigma_l^{(1)}} \ot \cdots \ot X_{\sigma_l^{(m_l)}}]   \nn\\
            &=  \tr(R_{\sigma_1}) \tr(R_{\sigma_2}) \cdots \tr(R_{\sigma_{l-1}}) \, R_{\sigma_l}\,,
    \end{align}    
where $\tr_{\sigma}$ denotes the partial trace over all elements in cycle $\sigma$.
The second relation can be shown in a similar way. This ends the proof.
\end{proof}

Taking a full trace
\begin{equation}\label{eq:complete_contraction}
 \tr[T(\pi^{-1}) X_1 \ot \cdots \ot X_k] = \tr(R_{\sigma_1}) \tr(R_{\sigma_2}) \dots \tr(R_{\sigma_l})\,,
\end{equation}
one arrives at a product of traces ---  that is, a {\em polynomial invariant}.
For multipartite systems more general tensor contractions can be obtained in similar ways.
These then correspond to local unitary invariants~\cite{817508}.
However, these expression cannot always be written in terms of the basic matrix operations {\em trace, partial trace, partial transpose}, and {\em matrix multiplication} when $k\geq 4$~\cite{Szalay_2012}.
Then more general wiring diagrams as used for tensor networks can be useful~\cite{Bridgeman_2017}.

\section{The polarized Cayley-Hamilton map}\label{sect:pCH}
Here we construct a first example of a trace polynomial that is positive on the positive cone,
which equivalently can be understood as a multilinear equivariant positive map.
We show that this map is not only positive, 
but also 
equivariant under unitaries, 
completely copositive, 
and tensor-stable.
Interestingly, this trace polynomial inequality for the positive cone 
has the same form as a matrix {\em identity} found by Lew in 1966, 
emphasizing the point that {\em ``inequalities are not broken equations''}~\cite{Orlin2019}.

Recall that $S_k$ is the symmetric group and $T$ the unitary representation that permutes the tensor factors from 
$(\C^d)^{\otimes k}$.
Let $\lambda$ be a partition of $k$, $\chi_\lambda$ be its character, 
and $P_\lambda$ the associated central Young projector given by
\begin{equation}
 P_\lambda = \frac{\chi_\lambda(e)}{k!} \sum_{\pi \in S_k} \chi_\lambda(\pi^{-1}) T(\pi)\,.
\end{equation} 

\begin{definition}\label{def:polCH_map}
Let $X_1, \dots, X_{k-1}$ be complex $d \times d$ matrices.
Let $\lambda \vdash k$ be a partition and 
$P_\lambda$ be the corresponding central Young projector. 
We define the polarized Cayley-Hamilton map
$ f_\lambda: \MM_d^{k-1} \to \MM_d$ 
as
\begin{equation}\label{eq:polCH_map}
    f_\lambda(X_1, \dots, X_{k-1}) =
    \tr_{ 1\dots k \backslash k } \big[P_\lambda (X_1 \ot \dots \ot X_{k-1} \ot \one)\big]\,,
\end{equation}
where the trace is performed over all but the last tensor factors.
\end{definition}
A complete list of all non-trivial polarized Cayley-Hamilton maps up to $k=4$ can be found in Appendix~\ref{app:tables}.

\subsection{Some observations}
\begin{lemma}\label{lemma:helper_trace_identity}
For all $\ket{v}, \ket{w} \in \C^d$ and $X_1,\dots,X_{k-1} \in \MM_d$, it holds that 
 \begin{equation}
  \bra{w} f_\lambda(X_1, \dots, X_{k-1}) \ket{v}
 = \tr\big[P_\lambda (X_1 \ot \cdots \ot X_{k-1} \ot \ketbra{v}{w})\big]\,.
 \end{equation}  
\end{lemma}
\begin{proof}
We use the coordinate-free definition of the partial trace from Eq.~\eqref{eq:coord_free_def_ptrace}. 
  \begin{align}
    \bra{w} f_\lambda(X_1, \dots, X_{k-1}) \ket{v} 
   &= \tr\big\{\tr_{ 1\dots k \backslash k } [P_\lambda (X_1 \ot \dots \ot X_{k-1} \ot \one)]  \cdot  \ketbra{v}{w} \big\} \nn\\
   &= \tr\big\{ [P_\lambda (X_1 \ot \dots \ot X_{k-1} \ot \one)] (\one \ot \dots \ot \one \ot \ketbra{v}{w}) \big\}   \nn\\
   &= \tr\big[ P_\lambda (X_1 \ot \dots \ot X_{k-1} \ot \ketbra{v}{w}) \big] \,.
   \end{align}
  This ends the proof.
\end{proof}

We translate $f_\lambda$ into a trace polynomial.
Given only $k-1$ matrix variables $X_1, \dots, X_{k-1}$, 
we define 
$\tilde R_{\sigma} = R_{\sigma}(X_1, \dots, X_{k-1}, X_k = \one)$.

\begin{observation}\label{lemma:PCH_V2}
The map $f_\lambda$ can be written as
\begin{equation}\label{eq:PCH_V2}
f_\lambda(X_1, \dots, X_{k-1}) = 
\frac{\chi_\lambda(e)}{k!} \sum_{\pi \in S_k} \chi_\lambda(\pi^{-1}) 
\prod_{i=1}^{l-1} \tr(R_{\sigma_i^\pi}) \tilde{R}_{\sigma_l^\pi}\,.
\end{equation}
\end{observation}

\begin{proof}
    We use Corollary~\ref{cor:perm_and_ptrace} and the decomposition 
    $P_\lambda = \frac{\chi_\lambda(e)}{k!} \sum_{\pi \in S_k} \chi_\lambda(\pi^{-1}) T(\pi)$. Then
    \begin{align}
    \frac{\chi_\lambda(e)}{k!} \sum_{\pi \in S_k} \chi_\lambda(\pi^{-1}) 
    \prod_{i=1}^{l-1} \tr(R_{\sigma_i^\pi}) \tilde{R}_{\sigma_l^\pi} 
    &= \frac{\chi_\lambda(e)}{k!} \sum_{\pi \in S_k} \chi_\lambda(\pi^{-1})  \tr_{1\dots k \backslash k} \big[ T(\pi^{-1}) X_1\ot \dots \ot X_{k-1} \ot \one\big]\nonumber\\
    &= \tr_{1\dots k \backslash k}\big[ \frac{\chi_\lambda(e)}{k!} \sum_{\pi \in S_k} \chi_\lambda(\pi^{-1})  T(\pi) X_1\ot \dots \ot X_{k-1} \ot \one \big]\nonumber\\
    &= \tr_{1\dots k \backslash k}\big[P_\lambda (X_1\ot \dots \ot X_{k-1} \ot \one)\big]\,.
    \end{align}
    This ends the proof.
\end{proof}

We are ready to explore some interesting properties of $f_\lambda$.

\subsection{The polarized Cayley-Hamilton identity}
The following matrix identity was proven by Lew in 1966, for which we give a new proof.
\begin{theorem}[Polarized Cayley-Hamilton identity \cite{Lew1966}]\label{thm:polCH}
 Let $\lambda$ be a partition of $k$ with strictly more than $d$ parts (i.e. the associated tableau has strictly more than $d$ rows).
 For all $X_1, \dots, X_{k-1} \in \MM_d$ it holds that
\begin{equation}
    f_\lambda(X_1, \dots, X_{k-1}) = 0\,.
\end{equation} 
\end{theorem}
 \begin{proof}
    It follows from the Schur-Weyl Duality~[Theorem~\ref{thm:schur_weyl}] that
    $P_\lambda  \ket{\phi_1} \ot \cdots \ot \ket{\phi_k} = 0$ when
    $\lambda$ has strictly more than $d$ parts 
    (c.f Proposition~\ref{prop:pigeon} in Appendix~\ref{app:RepTheorySk}).
    Expanding the matrices $X_i$ in a vector basis, one has
    $P_\lambda X_1 \ot \cdots \ot X_{k-1} \ot \one = 0$, and consequently also 
    $f_\lambda(X_1, \dots, X_{k-1}) = \tr_{1\dots k \backslash k} [P_\lambda (X_1 \ot \cdots \ot X_{k-1} \ot \one)] = 0$.
    This ends the proof.
    \end{proof}

\begin{example}
The partition $\lambda=(1,1) \vdash 2$ yields the idempotent
$\omega_\lambda = \frac{1}{2}[() - (12)]$.
By evaluating Eq.~\eqref{eq:polCH_map} we obtain (after a scaling)
$f_\lambda(X) = \tr(X) - X$, in entanglement theory also known as the reduction map.
Theorem~\ref{thm:polCH} states the (trivial) fact that $f_\lambda$ vanishes on $1\times 1$ matrices.
\end{example}

Procesi and Razmyslov independently showed that {\em all} multilinear trace identities 
that hold for complex $d\times d$ matrices are consequences of the Cayley-Hamilton Theorem 
and that they are completely described by 
Young tableaux~\cite{Formanek1989,Razmyslov_1974,PROCESI1976306};
see Section~\ref{sect:PI} for more details.

\subsection{$f_\lambda$ is positive on the positive cone}
We now extend the polarized Cayley-Hamilton identity [Theorem~\ref{thm:polCH}]
to an inequality for the positive cone. 
We term it the {\em polarized Cayley-Hamilton inequality}.
However, 
in light of concepts known from quantum information theory 
(c.f. Section \ref{sect:pCHinQIT}),
one could equally see it as a multivariate generalization 
of the universal state inversion
or of the shadow operator inequality.

\begin{theorem}\label{thm:pCH_pos}
 The map $f_\lambda$ is positive on the positive cone.
 In other words, 
\begin{equation}
 f_\lambda(X_1, \dots, X_{k-1}) \geq 0 \quad \text{whenever} \quad X_1, \dots, X_{k-1} \geq 0\,.
\end{equation}
\end{theorem}
\begin{proof}
Recall that a matrix $X$ is positive semidefinite 
if and only if the expression $\bra{\phi} X \ket{\phi}$ is real and nonnegative 
for all $\ket{\phi} \in \C^d$~\cite[Theorem 4.1.4]{horn_johnson_2012}. 
 With Lemma~\ref{lemma:helper_trace_identity}, we obtain the nonnegative expression
\begin{equation}
    \bra{\phi} f_\lambda(X_1, \dots, X_{k-1}) \ket{\phi} 
    = \tr\big(P_\lambda X_1 \ot \cdots \ot X_{k-1} \ot \dyad{\phi}\big) \geq 0\,,
\end{equation} 
as latter is the Hilbert-Schmidt inner product of two positive semidefinite matrices.
This ends the proof.
\end{proof}
\begin{example}\label{ex:pCH_111}
 The partition $\lambda = (1,1,1) \vdash 3$ yields the idempotent
$\omega_\lambda= \frac{1}{6}[()-(12)-(13)-(23)+(123)+(132)]$ and 
 \begin{equation}
  f_\lambda(X,Y) = \frac{1}{6}[\tr(X)\tr(Y){} - \tr(XY){} - \tr(X)Y - \tr(Y)X + XY + YX]\,.
 \end{equation}
 The map $f_\lambda$ vanishes on $2\times 2$ matrices [Theorem~\ref{thm:polCH}] 
 and $f_\lambda(X,Y) \geq 0$ whenever $X, Y \geq 0$ 
 [Theorem~\ref{thm:pCH_pos}].
\end{example}

\subsection{$f_\lambda$ is equivariant under unitaries}
The map $f_\lambda$ arises as a lifting of polynomial invariants. 
Consequently $f_\lambda$ is equivariant under the simultaneous conjugate action of unitaries.

\begin{proposition}\label{prop:pCH_covariance}
  The map $f_\lambda$ is equivariant under the action of the unitary group, that is, 
  \begin{equation}
   f_\lambda(UX_1 U^{-1}, \dots, U X_{k-1} U^{-1}) = U f_\lambda(X_1, \dots, X_{k-1}) U^{-1} \quad \text{for all} \quad U \in \UU\,.
  \end{equation} 
  \end{proposition}
\begin{proof}
Because $f_\lambda$ is hermitian~\footnote{
If $X_1, \dots, X_{k-1}$ is not hermitian, 
then $f_\lambda$ is not necessarily hermitian either. 
However, the equivariance of $f_\lambda$ still holds.},
it is enough to show that 
  \begin{equation}\label{eq:equiv_to_show}
   \tr\big[ U^{-1} f_\lambda(UX_1 U^{-1}, \dots, U X_{k-1} U^{-1}) U \dyad{\phi}\big] 
   = \tr\big[f_\lambda(X_1, \dots, X_{k-1}) \dyad{\phi}\big] \,.
  \end{equation} 
holds for all $\dyad{\phi} \in \MM_d$.
Using Lemma~\ref{lemma:helper_trace_identity}, the cyclicity of the trace, and the Schur-Weyl duality [Theorem~\ref{thm:schur_weyl}] we write
  \begin{align}
     \tr\big[ U^{-1} f_\lambda(UX_1 U^{-1}, \dots, U X_{k-1} U^{-1}) U \cdot \dyad{\phi}\big] 
   = &\tr\big[ P_\lambda  \big( U X_1 U^{-1} \ot \cdots \ot U X_{k-1} U^{-1} \ot U \dyad{\phi} U^{-1}\big)\big] \nn\\
   = &\tr\big[ P_\lambda  U^{\otimes k}
                    \big(X_1 \ot \cdots \ot X_{k-1} \ot \dyad{\phi} \big) 
                    (U^{-1})^{\otimes k} \big] \nn\\
   = &\tr\big[ P_\lambda  (X_1 \ot \cdots \ot X_{k-1} \ot \dyad{\phi} ) \big] \nn\\
   = &\tr\big[ f_\lambda (X_1, \dots, X_{k-1}) \dyad{\phi} \big] \,.
  \end{align}
This ends the proof.
\end{proof}
Of course, the same proof establishes that 
more generally, 
$f_\lambda$ is equivariant under the action of $A^{\otimes k}$ with $A \in GL(\C^d)$.

\subsection{$f_\lambda$ is tensor-stable}
How does the map $f_\lambda$ behave under the tensor product?
To interpret expressions such as $f_\lambda^{\ot 3}$ or $f_\lambda \ot f_\mu$
we return to Definition~\ref{def:polCH_map},
\begin{equation}
    f_\lambda(X_1, \dots, X_{k-1}) =
    \tr_{ 1\dots k \backslash k } \big[P_\lambda (X_1 \ot \cdots \ot X_{k-1} \ot \one)\big]\,.
\end{equation}
It should now be clear how to define tensor products of this map.
Let $X_1, \dots, X_{k-1}$ be operators acting on $\C^d \ot \C^{d'}$ and 
let $\lambda$ and $\mu$ be (possibly distinct) partitions of $k$.
We define the tensor product of $f_\lambda$ and $f_\mu$ as
\begin{equation}
 f_\lambda \ot f_\mu(X_1, \dots, X_{k-1}) = \tr_{ 1 1'\dots k k' \backslash \{kk'\} } \big[P_{\lambda\mu}(X_{11'} \ot \cdots \ot X_{(k-1) (k-1)'} \ot \one_{d d'})\big]\,,
\end{equation}
where $ P_{\lambda\mu} = P_\lambda \ot P_\mu$ 
is a ''vertical`` tensor product, with $P_\lambda \in (M_{d})^{\ot k}$ and $P_\mu \in {M_{d'}}^{\ot k}$ 
acting on parties $1\dots k$ and $1' \dots k'$ respectively. 
The expression is positive on the positive cone because as in the proof of Theorem~\ref{thm:pCH_pos}, 
\begin{equation}
 \bra{\phi} f_\lambda \ot f_\mu(X_1, \dots, X_{k-1}) \ket{\phi} \geq 0
\end{equation} 
holds for all $\ket{\phi} \in \C^{d d'}$.
We arrive at the following result.
\begin{theorem}\label{thm:pCH_is_tensor_stable}
The map $f_\lambda$ is tensor-stable.
That is, 
for any choice of $n$ partitions $\lambda, \dots, \mu \vdash k$ and 
all $X_1, \dots, X_{k-1} \in M_{d}^+$, the expression
\begin{equation}
 f_{\lambda} \ot \cdots \ot f_{\mu} (X_1, \dots, X_{k-1}) \geq 0
\end{equation} 
is positive semidefinite on the positive cone of $d \times d$ matrices ($d = d_1 \cdots d_n$).
\end{theorem}

We now state two examples with density matrices as variables. 
Then many normalization factors fall away.
\begin{example}\label{ex:shadow_op_exmpl}
The partitions $(1,1) \vdash 2$ and $(2) \vdash 2$ 
yield the idempotents
 $\omega_{-}   = \frac{1}{2}\big[() - (12)\big]$ and 
 $\omega_{+}     = \frac{1}{2}\big[() + (12)\big]$.
 The maps
$f_{-}(X) = \tr(X){} - X$ and 
$f_{+}  (X) = \tr(X){} + X$ follow. 
Let $T$ be a subset of $\{1 \dots n\}$ and 
choose for every $j$ the partition $(1,1)$ if $j\in T$ and the partition $(2)$ otherwise.
The positivity of $\bigotimes_{j \in T} f_{-}^{(j)} \bigotimes_{i \not \in T} f_{+}^{(i)}$
yields an inequality for multipartite quantum states: 
let $\rho$ be a density matrix acting on $\C^{d_1} \ot \cdots \ot \C^{d_n}$.
Then for all subsets $T \subseteq \{1\dots n\}$,
 \begin{equation}\label{eq:shadow}
        \sum_{S \subseteq \{1 \dots n\}} (-1)^{|S\cap T|} \tr_{S^c}(\rho) \ot \one_{S^c} \geq 0\,.
 \end{equation}
 Above inequality is also known as the generalized universal state inversion 
 or shadow operator inequality [c.f. Sections~\ref{sect:related}, \ref{sect:pCHinQIT}]. 
\end{example}

\begin{example}\label{eq:21_tensorstab}
The partition $\lambda = (2,1) \vdash 3$ yields the idempotent 
$\omega_{\lambda}= \frac{1}{3}\big[2() - (123) - (132)\big]$.
The map
$f_{\lambda}(X,Y) = \frac{1}{3} \big [ 2\tr(X) \tr(Y) {} - XY - YX \big]$ follows, 
whose wiring diagram is visualized in Figure~\ref{fig:lambda_21}.
\begin{figure}[tb]
 \includegraphics{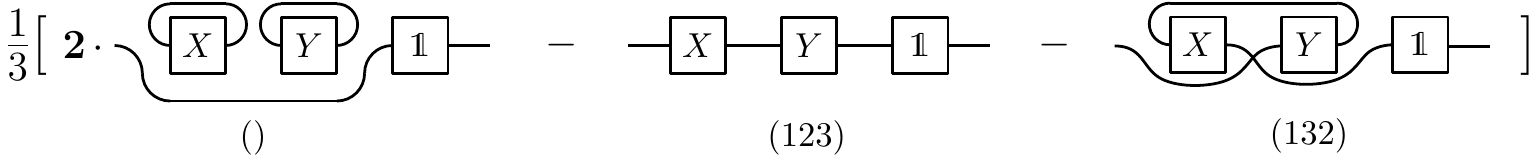}
 \caption{\label{fig:lambda_21}
 The wiring diagram of the trace polynomial corresponding to 
 $\frac{1}{3}\big[2() - (123) - (132)\big] \in \C S_3$.}
\end{figure}
Consider Theorem~\ref{thm:pCH_is_tensor_stable} for a bipartite system first: 
let $\rho$ and $\mu$ be density matrices on $\C^{d_1} \ot \C^{d_2}$. Then
\begin{align}\label{eq:21_tensorstab_2party}
 f_\lambda^{\ot 2}(\rho, \mu) = \frac{1}{9}\big[ &4(\one \ot \one) 
 + (\rho^{T_1} \mu^{T_1})^{T_1}
+ (\rho^{T_2} \mu^{T_2})^{T_2} 
- 2 (\{\rho_1, \mu_1\} \ot \one + \one \ot \{\rho_2, \mu_2\})
+ \{\rho, \mu \}
\big]
\geq 0 \,,
\end{align}
where $\rho_1 = \tr_2(\rho)$ and $\rho^{T_1}$ etc denote the reduced density matrix and the partial transpose of $\rho$ respectively.

The expression for $n$-partite systems can symbolically be expanded as
\begin{equation}\label{eq:lambda_21_symbolic}
 f^{\ot n}_\lambda(\#1, \#2) = \frac{1}{3^n} \bigotimes_{j \in \{1\dots n\}} 
 \Big[ 
 2\tr_j(\# 1) \tr_j(\# 2) \one 
 - \id(\#1) \cdot \id(\#2) 
 - \theta \big(\theta(\#1) \cdot \theta(\#2) \big ) 
 \Big] \,.
\end{equation} 
Above, $\id$ and $\theta$ are the identity and transpose maps, and $\#1$ and $\#2$ refer to performing the operation on the first and second variable respectively.
Let $\rho, \mu$ be density matrices on $\C^{d_1} \ot \dots \ot \C^{d_n}$ 
and denote by $\rho_A = \tr_{A^c}(\rho)$ and $\rho^{T_A}$ the reduction and the partial transpose on subsystem~$A$ respectively.
Theorem~\ref{thm:pCH_is_tensor_stable} then states that
\begin{equation}\label{eq:deg_3}
 f^{\ot n}_\lambda (\rho, \mu) = (\tfrac{2}{3})^n
    \sum_{\substack{S \subseteq \{1\dots n\} \\ A \subseteq S}}  (-\tfrac{1}{2})^{|S|}
    (\rho_S^{T_A} \cdot 
     \mu_S^{T_A})^{T_A} \ot \one_{S^c} \geq 0\,.
\end{equation}
\end{example}

\begin{remark}
We emphasize that the above constructions in 
Examples~\ref{ex:shadow_op_exmpl} to \ref{eq:21_tensorstab}
were chosen because of their low degree.
The tensor-stability as established in 
Theorem~\ref{thm:pCH_is_tensor_stable} allows 
to combine {\em any} set of $n$ Young tableaux with $k$ boxes or less
to obtain multipartite trace polynomials in $k-1$ variables 
that are positive on the positive cone.
\end{remark}

\subsection{$f_\lambda$ is completely copositive}
Recall that a map $\Lambda$ is completely copositive if 
$\Lambda \circ \theta$ is completely positive
where $\theta$ is the partial transposition.
The partition $\lambda=(1,1) \vdash 2$ yields the map $f_{(1,1)}(X) = \tr(X)\one - X$,
which from quantum information theory is known to be completely copositive~\cite{PhysRevA.59.4206}.
We now show that this property holds for all $f_\lambda$.

Recall the definition 
$
 f_\lambda(X_1, \dots, X_{k-1}) = \tr_{1\dots k \backslash k}[P_\lambda (X_1 \ot \dots \ot X_{k-1} \ot \one)]\,.
$
This suggests that $P_\lambda$ plays a role similar to that of $\rho_\Lambda$ in 
the Choi-Jamiołkowski isomorphism [Eq.~\eqref{eq:CJ_inverse}].
This is indeed the case.
\begin{proposition}\label{prop:pCH_copos}
 The map $f_\lambda$ is completely copositive. That is, $f_\lambda$ is completely positive 
 when all its variables are partially transposed,
 \begin{equation}
  \big (f_\lambda \circ \theta \ot \id_d \big) (X_1, \dots, X_{k-1}) \geq 0 
  \quad\quad \text{if } \quad\quad  X_1, \dots, X_{k-1} \in (\MM_d \ot \MM_d)^+\,.
 \end{equation} 
 where $\theta$ acts on each variable.
\end{proposition}
\begin{proof}
Let $X \in \MM_d^{\otimes k-1}$ and consider the linear extension of $f_\lambda$,
\begin{equation}
 \tilde f_\lambda(X) =
 \tr_{1 \dots k \backslash k}  [P_\lambda (X \ot \one) ] 
 = d^{2-k} \tr_{1 \dots (2k-2)} 
 [(P_\lambda \ot \one_{d^{k-1}}) (X \ot \one_{d^{k-1}})] \,.
 \end{equation} 
With the Choi-Jamiołkowski isomorphism [c.f. Eq.~\eqref{eq:CJ_inverse}] 
it is clear that $\tilde f_\lambda$ is completely copositive.
This property still holds when $X$ is of tensor-product form 
$X = X_1 \ot \dots \ot X_{k-1}$.
It follows that $f_\lambda$ is completely copositive and 
this ends the proof.
\end{proof}

\subsection{Asymmetric tensor-stable positive maps}
\label{sect:asym}
So far the maps considered were symmetric in the variables. 
Considering single copies of irreducible representations, 
one can obtain a further fine-graining of the symmetric Cayley-Hamilton map 
into asymmetric summands.

Let us state an example. The partition $(2,1)\vdash 3$ 
yields the hermitian idempotent $\omega_\lambda = \frac{1}{3}[2() - (123) - (132)]$.
It is centrally primitive and corresponds to the isotypic component with projector 
$P_\lambda = \hat T(\omega_\lambda)$.
Decomposing $\omega_\lambda$ further, one can obtain the two orthogonal primitive hermitian idempotents
$\frac{1}{6}[2() + 2(12) - (23) - (13) - (123) - (132)]$ and 
$\frac{1}{6}[2() - 2(12) + (23) + (13) - (123) - (132)]$. 
Another decomposition is
$\frac{1}{6}[2() + 2(23) - (12) - (13) - (123) - (132)]$ and 
$\frac{1}{6}[2() - 2(23) + (12) + (13) - (123) - (132)]$.
It is important that these idempotents yield {\em hermitian} projectors under the algebra representation $\hat T$.
We now replace $P_\lambda$ in Eq.~\eqref{eq:polCH_map} 
with these fine-grained projectors and obtain the following four 
inequalities,

\begin{align}
\begin{rcases}
 &2\tr(X)\tr(Y){} + 2\tr(XY){} - \tr(X)Y - \tr(Y)X - XY - YX \\
 &2\tr(X)\tr(Y){} - 2\tr(XY){} + \tr(X)Y + \tr(Y)X - XY - YX \\
 &2\tr(X)\tr(Y){} + 2\tr(X)Y - \tr(XY){} - \tr(Y)X - XY - YX \\
 &2\tr(X)\tr(Y){} - 2\tr(X)Y + \tr(XY){} + \tr(Y)X - XY - YX
 \quad \end{rcases}
 \quad \geq 0 \quad \text{whenever}\quad X,Y \geq 0\,.
\end{align}
It is easy to see that the first two maps add up to $2\tr(X)\tr(Y) - XY - YX$, 
corresponding to the centrally primitive idempotent $2() - (123) - (132)$. The same holds for the last two maps.

Such a fine-graining into maps that arise from individual irreducible representations 
is non-trivial but always possible~\cite{
doi:10.1063/1.4983478}. 
It is important to note that it does not simply suffice to take idempotents in $\C S_k$,
but that they also have to be hermitian for the resulting map to be positive on the positive cone.
A generic construction can be obtained from hermitian sum-of-squares in the group ring, i.e. elements 
$\a^* \a$ with $\a \in \C S_k$~\cite{Huber2020b}.

\subsection{The polarized Cayley-Hamilton map in quantum information: the case $k=2$}\label{sect:pCHinQIT}
In quantum information, the reduction map $\RR(\rho) = f_\lambda(\rho) = \rho \mapsto \one - \rho$ 
with $\lambda=(1,1)\vdash 2$ stands as the archetypical example for a class of univariate maps that are 
positive but not completely~\cite{
PhysRevLett.97.080501,
0305-4470-39-45-020,
doi:10.1063/1.4962339}.
It arose $1997$ as an criterion in entanglement detection and distillation:
whenever
\begin{equation}
 (\RR \ot \id)(\rho_{AB}) = \one_1 \ot \rho_B - \rho_{AB} \ngeq 0\,,
\end{equation}
then $\rho_{AB}$ is entangled and multiple copies of $\rho_{AB}$ can be distilled 
to the maximally entangled state~\cite{
PhysRevA.59.4206,
PhysRevA.60.898}.

The complete copositivity of the reduction map immediately leads to the quantum channel
$\rho \mapsto \frac{1}{d-1}(\one - \rho^T)$.
Known as the {\em Werner-Holevo channel}, it played a pivotal role 
in refuting the maximal $p$-norm multiplicativity conjecture~\cite{doi:10.1063/1.1498491, HaydenWinter2008}.
Furthermore, the tensor-stability of $\RR$ led to the development of the {\em universal state inversion}~\cite{
PhysRevA.64.042315,
PhysRevA.72.022311},
which was used to obtain compatibility conditions for the quantum marginal problem~\cite{
Butterley2006, 
PhysRevA.75.032102}
and monogamy of entanglement relations~\cite{
Eltschka2018distributionof, 
PhysRevA.98.052317,
Huber_Thesis}. 

Independently, Rains generalized the shadow bounds 
from classical error correction to the quantum setting. 
He termed the resulting trace inequality the {\em shadow inequality} (preprint in $1996$, published in $1999$~\cite{796376}):
for all operators $M,N \geq 0$ in 
$\MM_d^{\otimes n}$ 
and 
$T \subseteq \{1\dots n\}$, one has that
\begin{equation}
 \sum_{S \subseteq \{1\dots n\}} (-1)^{|S \cap T|} \tr\big[\tr_{S^c}(M) \tr_{S^c}(N)\big] \geq 0\,.
\end{equation} 
where $S^c$ is the complement of $S$ in $\{1\dots n\}$. 
Note how the above expression can, with the coordinate-free definition of the partial trace,
be lifted to the shadow operator inequality from Example~\ref{ex:shadow_op_exmpl}.

Stated in terms of polynomial invariants of degree two, the shadow inequality can be 
incorporated as a list of additional constraints into the quantum linear programming bound~\cite{
681316, 
796376}.
Using an Ansatz for the weight distribution, Rains showed that qubit codes can have a distance 
of at most $d\lesssim  n/3$~\cite{651000}.
The shadow inequalities can also be evaluated directly if the explicit 
form of the weight enumerator is known, 
leading to occasionally tight bounds on the existence of absolutely maximally entangled states 
and quantum maximum distance separable codes~\cite{
1751-8121-51-17-175301, 
Huber2020quantumcodesof}. 
We note that these shadow bounds are nothing else than the above-mentioned monogamy of 
entanglement relations applied to the case of quantum codes.

In a seminal article ''Polynomial invariants of quantum codes``, 
Rains obtained a multilinear generalization of the shadow inequality~\cite{817508}.
These generalized shadow inequalities were still in the form of {\em trace} inequalities
and are the original motivation for this article.

\section{Positive maps from entanglement witnesses}\label{sect:matrix_ineq}
It is well-known that positive but not completely positive maps correspond to entanglement witnesses.
Maps covariant with respect to certain operations were studied in Refs.~\cite{
Mozrzymas2017, 
BardetCollinsSapra2020},
and multilinear maps in Refs.~\cite{
HORODECKI20011,
doi:10.1063/1.4931059,
Kye_2015}.

Here we relate the construction of equivariant positive maps 
(i.e. multilinear trace polynomial inequalities for the positive cone) 
to invariant block-positive operators and entanglement witnesses.
In contrast to the previous Section, 
these more general constructions generally lack properties such as 
tensor-stability and complete copositivity.

We recall that a multipartite quantum state is termed {\em separable} 
if it can be written as a convex combination of product states,
$
 \rho = \sum_i p_i \rho_1^{(i)} \ot \dots \ot \rho_k^{(i)}.
$
Here $\rho_j^{(i)}$ are quantum states and $p_i$ are probabilities 
such that $\sum_i p_i = 1$.
A quantum state which is not separable, 
i.e. for which no such decomposition can be found, 
is called {\em entangled}.
Suppose one has an operator $\mathcal{W}$ for which 
$\tr[\mathcal{W} \rho ] \geq 0$ 
holds for all separable states $\rho$, 
then $\mathcal{W}$ is {\em block-positive}.
If additionally $\tr[\mathcal{W} \sigma]<0$ for some entangled state 
$\sigma$, then $\mathcal{W}$ is termed an {\em entanglement witness}.
Consequently, $\mathcal{W}$~can be used to detect quantum entanglement.
Lastly, $\mathcal{W}$ is an {\em optimal} witness 
if also $\tr[\mathcal{W} \tilde\rho] = 0$ 
holds for some separable state~$\tilde\rho$.
For our purposes, we require entanglement witnesses that 
detect $U^{\ot k}$-invariant states. 
These so-called Werner states satisfy 
$\rho_{12} = (U\ot U) \rho_{12} (U^\dag \ot U^\dag)$,
which generalizes to the $k$-partite case as
$\rho_{1\dots k} = U^{\ot k} \rho_{1\dots k} \,(U^\dag)^{\ot k}$, for all $U \in \UU$.
The separability of tripartite Werner states was studied in 
Ref.~\cite{PhysRevA.63.042111} and only little is known in the case 
of more than three parties~\cite{MaassenKuemmerer2019}.
Due to linearity, Werner state witnesses can always 
assumed to show the same invariance as the states themselves, 
and $\WW = U^{\ot k} \WW \,(U^\dag)^{\ot k}$ for all $U\in\UU$ holds.

\bigskip
\noindent {\bf General construction.}
Let $\PP\in \MM_d^{\ot r+1}$ and define the map
 \begin{equation}\label{eq:gen_equiv_map}
  f_{\mathcal{P}}(X_1, \dots, X_r) = \tr_{1\dots r}\big[\mathcal{P} (X_1 \ot \dots \ot X_r \ot \one)\big]
 \end{equation}
 We use the coordinate-free definition of the partial trace [ Eq.~\eqref{eq:coord_free_def_ptrace}]
 to establish that
 \begin{equation}
  \bra{\phi} f_{\mathcal{P}}(X_1, \dots, X_r) \ket{\phi} = \tr \big[\mathcal{P} (X_1 \ot \dots \ot X_r \ot \dyad{\phi})\big]
 \end{equation}
 holds for all $\ket{\phi} \in \C^d$.
 
 Let now $X_1,\dots, X_r\geq 0$ be positive semidefinite. Then 
\begin{equation}\label{eq:pos_map_gen}
 f_{\mathcal{P}}(X_1, \dots, X_r) \geq 0 \quad \quad \text{if and only if} 
 \quad \quad \tr \big[\mathcal{P} (X_1 \ot \dots \ot X_r \ot \dyad{\phi})\big] \geq 0 
 \quad \text{for all } \ket{\phi}\in \C^d\,.
\end{equation} 
When the above expression is positive semidefinite then $\PP$ must be a block-positive operator, 
e.g. an entanglement witness.

\subsection{Trace polynomials}
Let us focus on the case of multilinear maps that can be realized as products and linear
combinations of matrix monomials $X_{\alpha_1} \cdots X_{\alpha_r}$ 
and their traces $\tr(X_{\alpha_1} \cdots X_{\alpha_r})$. 
So these {\em trace polynomials} have the general form of
\begin{equation}
  \sum_{\alpha, \beta} c_{\alpha \beta} X_{\alpha_1} \cdots X_{\alpha_r} \prod_\beta \tr \big(X_{\beta_1} \cdots X_{\beta_t}\big) 
\end{equation}
where $c_{\alpha\beta} \in \C$ and $\alpha, \beta$ are multi-indices.
A trace polynomial $f$ is an inequality on the positive cone, 
if $f(X_1, \dots, X_r) \geq 0$  for all $X_1, \dots, X_r \geq 0$ holds.
It is optimal if also 
$\lambda_{\min} \big(f(\tilde{X}_1, \dots, \tilde{X}_r)\big) = 0$ holds
for some nonzero $\tilde{X}_1, \dots, \tilde{X}_r \geq 0$,
where $\lambda_{\min}$ stands for the smallest eigenvalue.

We construct matrix inequalities from entanglement witnesses.
\begin{theorem}\label{thm:witness_to_ineq}
Every multilinear trace polynomial inequality on the positive cone in $r$ variables
corresponds to an $U^{\otimes r+1}$-invariant block-positive operator.
In particular, every optimal multilinear trace polynomial inequality on the positive cone 
corresponds to an optimal Werner state witness.
\end{theorem}

\begin{proof}
We first show that any multilinear trace polynomial inequality on the positive cone 
corresponds to an $U^{\otimes r+1}$-invariant block-positive operator.\\
\noindent ''$\Rightarrow$``:
let $f$ be a multilinear trace polynomial inequality on the positive cone in $r$ matrix variables.
Corollary~\ref{cor:perm_and_ptrace} allows us to express the trace polynomial as
\begin{equation}
 f(X_1, \dots, X_r) = 
 \tr_{1 \dots {r}}
 \big[\hat T(\a) (X_1 \ot \dots \ot X_r \ot \one)\big]  \quad \text{with} \quad  \a \in \C S_{r+1} \,. 
\end{equation}
By assumption, $f$ is positive semidefinite on the positive cone.
With Eq.~\eqref{eq:pos_map_gen}, this can be restated as the requirement that
\begin{equation}
 \bra{\phi} f(X_1, \dots, X_r) \ket{\phi} 
 = \tr \big[\hat T(\a) (X_1 \ot \dots \ot X_r \ot \dyad{\phi})\big] \geq 0
\end{equation}
holds for all $\ket{\phi} \in \C^d$ and $X_1, \dots, X_r \geq 0$.
Then $\hat T(\alpha)$ must be hermitian due to linearity~\cite{ChauDavid_privcomm2019}
and finally, also block-positive.
From the Schur-Weyl duality [Theorem~\ref{thm:schur_weyl}] it follows that 
$\hat T(\a)$ commutes with the diagonal action of unitaries and therefore 
$\hat T(\a) = U^{\ot (r+1)} T(\a) (U^\dag)^{\ot (r+1)}$.
Thus $\hat T(\a)$ is $U^{\ot (r+1)}$-invariant block-positive operator.
\\
\noindent ''$\Leftarrow$``:
conversely, let $\mathcal{B} \in \MM_d^{\ot r+1}$ 
be a block-positive operator with the property that
$\mathcal{B} = U^{\ot (r+1)} \, \mathcal{B} \, (U^\dag)^{\ot (r+1)}$
for all unitaries $U \in \mathcal{U}_d$. 
From the Schur-Weyl duality [Theorem~\ref{thm:schur_weyl}] it follows that 
$\mathcal{B}$ can be decomposed into a linear combination of permutations.
Consequently the expression
\begin{equation}
f_\mathcal{B}(X_1, \dots, X_r) = \tr_{1\dots r}
\big[\mathcal{B} (X_1\ot \dots \ot X_{r} \ot \one)\big]
\end{equation} 
constitutes a multilinear trace polynomial inequality on the positive cone.
[We again used the coordinate free definition of the partial trace 
or Eq.~\eqref{eq:pos_map_gen} respectively.]

We now show that $f_\mathcal{W}(X_1, \dots, X_r) =  \tr_{\{1\dots r\}}
\big[\mathcal{W} (X_1\ot \dots \ot X_{r} \ot \one)\big]$ is an {\em optimal} trace polynomial inequality
if and only if
$\mathcal{W}$ is an {\em optimal} $U^{\otimes r+1}$-invariant entanglement witness.
Because $f_\mathcal{W}$ is multilinear, we can always normalize the variables $X_i \in \MM_d^+$
to quantum states $\rho_i = X_i / \tr(X_i)$.
Consider the identity 
\begin{equation}\label{eq:optimal}
\sum_i p_i  \bra{{\phi_i}} f_\mathcal{W} (\rho_1, \dots, \rho_{r}) \ket{{\phi_i}} =
\tr\big[\mathcal{W} ( \rho_1 \ot \dots \ot \rho_r \ot \sum_i p_i \dyad{\phi_i})\big]\,.
\end{equation}
where $p_i \geq 0 $ and $\sum_i p_i = 1$ such that $\sum_i p_i \dyad{\phi_i}$ is a density matrix.
Due to linearity, the separable state reaching $\tr(\mathcal{W} \rho_{\text{sep}}) = 0$ in Eq.~\eqref{eq:optimal} 
can always taken to be pure.
It is now easy to see that an optimal trace polynomial inequality yields an optimal witness, and vice versa.
This ends the proof.
\end{proof}

\subsection{Equivariant positive maps}
One can readily see that all multilinear trace polynomials are equivariant under unitary action,
\begin{equation}
 f_\BB(UX_1 U^{-1}, \dots, U X_r U^{-1}) = U f_\BB( X_1 , \dots, X_r )U^{-1} \quad \text{for all} \quad U\in \UU\,.
\end{equation}
It is known that the ring of equivariant maps is generated over the ring of invariants 
of the form $\tr(X_{\alpha_1} \cdots X_{\alpha_r})$ by matrix monomials (trace polynomials)~\cite{ConciniProcesi2017}.
Thus the set of multilinear equivariant maps and the set of multilinear trace polynomials coincide 
(see also Section~\ref{sect:PI}).
We can thus readily rephrase Theorem~\ref{thm:witness_to_ineq} as the following.

\begin{corollary}\label{cor:ineq_to_witness}
Every multilinear positive map that is equivariant under unitaries
corresponds to a unitary-invariant block-positive operator.
In particular, every optimal multilinear positive map that is equivariant under unitaries
corresponds to an optimal Werner state witness.
\end{corollary}

\subsection{Some optimal trace polynomials}
The Cayley-Hamilton map $f_\lambda$ is a basic example of a trace polynomial inequality that is {\em not} optimal when $d\geq 2$.
This follows from the fact that $P_\lambda$ is a positive semidefinite operator, but not an entanglement witness.
Corollary~\ref{cor:ineq_to_witness} however allows to obtain a larger class of inequalities from Werner state witnesses.
Our next example corresponds to an optimal witness for tripartite Werner states,
and is in a sense complementary to Example~\ref{ex:pCH_111}.
\begin{example}\label{ex:werner}
Consider the hermitian idempotent
 $\omega_{-} = \frac{1}{6}[()-(12)-(13)-(23)+(123)+(132)]$ which corresponds to the partition $(1,1,1) \vdash 3$.
 A classic result by Eggeling and Werner established that~\cite[Theorem~$1$]{PhysRevA.63.042111}
 \begin{equation}
  \max_{\rho \in \text{SEP}} \tr[\hat{T}(\omega_{-}) \rho] \leq 1/6\,,
 \end{equation} 
 where the maximum is taken over the set of separable states $\text{SEP}$.
We construct the optimal witness $\mathcal{W} = \frac{1}{6}\one - \hat{T}(\omega_{-})$. 
Returning to unnormalized variables $X,Y \in \MM_d^+$, one obtains
\begin{align}\label{eq:opt_111}
 f_{\mathcal{W}}(X,Y) &= \tr_{12}[\mathcal{W} (X \ot Y \ot \one)]
 = \tr(XY){}  + \tr(X)Y + \tr(Y)X - XY - YX \geq 0 \quad \text{whenever} \quad X,Y\geq 0
\end{align}
by Theorem~\ref{thm:witness_to_ineq}. 
Also, $f_{\mathcal{W}}$ is an optimal trace polynomial inequality or stated differently, 
$f_\lambda$ is an optimal equivariant positive map. 
Indeed $\lambda_{\min}\big(f_{\mathcal{W}}(X,Y)\big) = 0$
whenever $X = \dyad{\psi}$ and $Y = \dyad{\psi^\perp}$ are orthogonal rank one operators
in $\MM_d^+$ with $d\geq 3$.
\end{example}
\begin{remark}
 Compare the inequality in Eq.~\eqref{eq:opt_111} to the inequality from Example~\ref{ex:pCH_111}. There we
 showed that
 \begin{equation}\label{eq:f_111again}
    f_\lambda(X,Y) = \tr(X)\tr(Y){} - \tr(XY){} - \tr(X)Y - \tr(Y)X + XY + YX \geq 0 
    \quad \text{whenever} \quad X,Y \geq 0\,.
 \end{equation}
In comparison to Eq.~\ref{eq:opt_111}, this expression is missing the first term $\tr(X)\tr(Y)$ 
with all signs inverted.
\end{remark}

The construction from Example~\ref{ex:werner} generalizes: 
let $\mathcal{W} = \frac{1}{k!}\one - T(\omega_{-})$, 
where $\omega_{-}$ is the central idempotent corresponding to the partition $(1,\dots,1)\vdash k$~\cite{MaassenKuemmerer2019}. 
This yields the trace polynomial
\begin{align}\label{eq:werner_antisym}
 f_{\mathcal{W}}(X_1, \dots, X_{k-1})  
 &= - \tr_{1 \dots k \backslash k} 
 \Big [ \sum_{\pi \in S_k, \pi \neq e} \sgn(\pi) T(\pi) (X_1 \ot \dots \ot X_{k-1} \ot \one ) \Big]  \nn\\
 &= \prod_{i=1}^{k-1} \tr(X_i) \one - f_{\lambda_{-}}(X_1, \dots, X_{k-1})
\,.
\end{align}
\begin{corollary}\label{cor:werner_antisym}
 The map $f_{\mathcal{W}}$ from Eq.~\eqref{eq:werner_antisym} is an optimal trace polynomial inequality on the positive cone.
\end{corollary}
\begin{proof}
 The claim follows from Theorem~\ref{thm:witness_to_ineq} and the result by Maassen and Kümmerer~\cite{MaassenKuemmerer2019} on optimal witnesses for symmetric Werner states.
\end{proof}
In the same work, Maassen and Kümmerer showed that the set of symmetric Werner states has 
an infinite number of extreme points for $k\geq 5$~\cite{MaassenKuemmerer2019}.
We obtain as a consequence the following.
\begin{corollary}
 The set of 
  symmetric multilinear trace polynomials positive on the positive cone 
 (symmetric multilinear equivariant positive maps) 
 in four or more matrix variables has an infinite number of extremal points.
\end{corollary}

The inequalities from entanglement witnesses are generally tighter than those originating 
from merely positive semidefinite operators. But, there is a price to pay:
first, their exact form {\em might} depend on the dimension $d$; 
and second, the corresponding multilinear positive maps are
not guaranteed to be tensor-stable.

\subsection{Positive trace polynomials that are not SOS}\label{sect:Motzkin}
For commutative variables, it is well-known that there exist polynomials 
that are positive on $\R$, but that do not have a {\em sum-of-squares} (SOS) form.
That is, they cannot be written as $\sum_i p_i q(x_1, \dots, x_k)^2$ 
where $q_i$ are polynomials and $p_i \geq 0$. 
A counterexample was found by Motzkin: the polynomial
$ M(x,y) = x^4 y^2 + x^2 y^4 + 1 - 3 x^2 y^2
$
is positive for all $x,y \in \R$ 
but cannot be brought into an SOS form~\cite{2d6a361a402f4d429d35e0a94b9ac8c4}.
This is in contrast to the non-commutative case, where
Helton showed that all symmetric non-commutative polynomials 
that are matrix-positive must be sums of hermitian squares~\cite{10.2307/3597203}.

Can a larger class of polynomials be obtained 
if one demands their positivity on the set of density matrices only?
Consider a non-commutative version of the Motzkin polynomial,
\begin{equation}\label{eq:Motzkin-type}
 M(A,B) = A B^4 A + A^2 B^2 A^2 - 3 A B^2 A + \one\,.
\end{equation} 
Numerical tests suggest that 
$M(A,B) \geq 0$ whenever $A,B \geq 0$ with $\tr(A) = \tr(B) = 1$.
A cyclicly equivalent polynomial appeared in Ref.~\cite[Example 4.4]{KLEP20081816} 
and is known to have a nonnegative trace on the set of hermitian matrices.

\begin{remark}
Since the first version of this article appeared, 
Jurij Volčič kindly communicated a proof of the Motzkin matrix inequality. 
One writes
\begin{align}
 & AB^4 A + A^2 B^2 A^2 - 3 A B^2 A + \one  \nn\\
 &= (A^2 - A) B^2 (A^2 - A) + (A B^2 + A^2 - 2 A)(B^2 A + A^2 - 2 A) + (\one - A)(\one + 2 A - A^2)(\one - A)\,.
\end{align}  
This shows that $M(A,B)$ is positive semidefinite 
when $(1-\sqrt{2}) \leq A \leq (1+\sqrt{2})$. 
In particular, this is the case for $A\geq 0$ and $\tr(A)=1$. 
\end{remark}

We can homogenize Eq.~\eqref{eq:Motzkin-type} 
with factors of $\tr(A)$ and $\tr(B)$, yielding
\begin{equation}
 \tilde M(A,B) = \tr(A)^2 A B^4 A + \tr(B)^2 A^2 B^2 A^2 - 3 \tr(A)^2 \tr(B)^2 A B^2 A 
 + \tr(A)^4\tr(B)^4\one\,.
\end{equation}
This removes the tracial constraints on $A$ and $B$.
Consequentially $\tilde M(A,B) \geq 0$ for all $A,B \geq 0$.
This shows that there are trace polynomials
that are positive on the positive cone but which are not SOS.

\section{tensor polynomial identities and invariant theory}
\label{sect:PI}

We first relate the methods developed in the previous sections 
to the theory polynomial identity rings.
A nice exposition of this topic is the article by Formanek~\cite{Formanek1989}.

A {\em polynomial invariant} is a multilinear function 
in the entries of $k$ matrices that is invariant under the simultaneous conjugate action by 
invertible matrices~\cite{ConciniProcesi2017},
\begin{align}
 &\iota: \MM_d \times \dots \times \MM_d \to \C\,, \quad \text{such that} \nn\\
 &\iota(A X_1 A^{-1}, \dots, A X_k A^{-1}) = \iota(X_1, \dots, X_k) \quad \text{ for all} \quad A \in GL_d(\C)\,.
\end{align} 
A closely related concept is that of {\em equivariant maps}. 
These are multilinear maps that are equivariant under the same action,
\begin{align}
 &f: \MM_d \times \dots \times \MM_d \to \MM_d\,, \quad \text{such that} \nn\\
 &f(A X_1 A^{-1}, \dots, A X_k A^{-1}) = A f(X_1, \dots, X_k) A^{-1} \quad \text{ for all} \quad A \in GL_d(\C)\,,
\end{align}
The difference here is that polynomial invariants yield scalars while equivariant maps yield matrices.
The terminology of {\em polynomial} invariants of degree $k$ for a matrix $X$ arises from choosing 
$X_i = X$ for all $1\leq i \leq k$.

The first fundamental theorem of matrix invariants states that the ring of polynomial invariants 
is generated by the traces of matrix monomials, i.e. elements of the form $\tr(X_{i_1} \cdots X_{i_r})$~\cite{ConciniProcesi2017}.
Likewise, the ring of equivariant maps is generated, over the ring of matrix invariants, by matrix monomials 
of the form $X_{i_1} \cdots X_{i_r}$.
From the fact that the Hilbert-Schmidt inner product is nondegenerate, 
it can be shown
that every polynomial invariant $\iota$ in $k+1$ variables is related to an 
equivariant map $f$ in $k$ variables by
\begin{equation}
 \iota(X_1, \dots, X_k, X_{k+1}) = \tr\big[  f(X_1, \dots, X_k) X_{k+1} \big] \,.
\end{equation}
Clearly, polynomial invariants and equivariant maps are connected to our formalism by 
\begin{align}
X_{i_1} \cdots X_{i_r}\,  &=  \tr_{1 \dots r \backslash r} \big[T((i_r \dots i_1)) X_{i_1} \ot \dots \ot X_{i_r}\big]\nn\,,\\
 \tr(X_{i_1} \cdots X_{i_r})   &= \quad \quad \,\,\tr \big[T((i_r \dots i_1)) X_{i_1} \ot \dots \ot X_{i_r}\big]\,.
\end{align} 
The above discussion establishes that the set of equivariant maps and the set of trace polynomials coincide.

\subsection{Trace and polynomial identities}
Consider now {\em identities} for polynomial invariants and equivariant maps.
Suitable linear combinations of polynomial invariants that are identically 
zero on $\MM_d$ are known as {\em trace identities}.
They are governed by the second fundamental theorem of matrix invariants.

\begin{theorem}[Procesi~\cite{PROCESI1976306}, Razmyslov~\cite{Razmyslov_1974}]
\label{thm:ProcesiRazmyslov}
 The expression $\tr[ \hat T(\alpha) X_1 \ot \dots \ot X_k]$ is a 
 multilinear trace identity on $\MM_d$ if and only if $\a \in \C S_k$ belongs 
 to the ideal that corresponds to partitions $\lambda \vdash k$ with more than $d$ parts.
 \end{theorem}
 This ideal is generated by the element $\epsilon = \sum_{\pi \in S_{d+1}} \sgn(\pi)\pi$
 and leads to the fundamental trace identity $\tr\big[T(\epsilon) X_1 \ot \dots \ot X_k \big]$
 which can be shown to arise from a linearization (also known as {\em polarization}) 
 of the Cayley-Hamilton Theorem~\cite{DrenskyFormanek2004}. 
 In other words, all multilinear trace identities are consequences of the Cayley-Hamilton Theorem.
 
A special type of identities for equivariant maps are {\em polynomial identities}: 
these are polynomials in non-commutative variables that vanish on the ring $\MM_d$ of complex $d \times d$ matrices.
So one has
\begin{equation}
 p(X_1, \dots, X_r) = 0 \in \MM_d \qquad \text{if} \qquad X_1, \dots, X_r \in \MM_d\,.
\end{equation} 
Note that these polynomials do not evaluate to a scalar zero, but to the zero matrix.
The arguably best known example is the so-called {\em standard polynomial}. It is defined as
\begin{equation}
 s_r(X_1, \dots, X_r) = \sum_{\pi \in S_r} \sgn(\pi) X_{\pi(1)} \cdots X_{\pi(r)}\,.
\end{equation}
A theorem by Amitsur and Levitzki states that $\MM_d$ satisfies the standard identity $s_{2d}$ in $2d$ variables
and that $\MM_d$ does not satisfy, up to multiplicative constants, 
any other polynomial identity of equal or lower degree~\cite{Amitsur1950, DrenskyFormanek2004}.
Given some matrix algebra, the required degree for the existence of a polynomial identity 
can be seen as a characterization of its non-commutativity~\cite{Kostant2009}.

A closely related concept is that of {\em central polynomials}.
These yield an non-zero element from the center $C(\MM_d)$, 
i.e. they evaluate to a scalar multiple of the identity matrix~$\one$.
For example, for all $2 \times 2$ matrices $A,B,C,D$ it holds that
\begin{align}
 [A, B][C, D] + [C, D][A, B] \,\, \propto \,\one\,,
\end{align}
where $[A,B] = AB - BA$ is the commutator~\cite{DrenskyFormanek2004}.
 
In quantum information, polynomial identities and central polynomials found use 
as bond dimension witnesses and cut-and-glue operators for matrix product states,
for manipulating the time evolution of quantum states,
and in the context of dimensional constraints in semidefinite programming hierarchies~\cite{
Navascues2018bonddimension, 
PhysRevX.8.031008, 
PhysRevA.92.042117}.

\subsection{Tensor polynomial identities}
The theme carries over to expressions on tensor product spaces. 
Here {\em swap polynomials} were introduced in the context of 
quantum remote time manipulation~\cite{Trillo2019}. 
These are polynomials on a bipartite tensor product space
of the form $\sum_i p_i \ot q_i$,
with $p_i$ and $q_i$ non-commutative polynomials,
that yield a scalar multiple of the swap operator~$\Gamma$.
More general, we consider {\em permutation polynomials} and 
{\em tensor polynomial identities} 
that 
evaluate to scalar multiples of some permutation operator, 
to scalar multiples of the identity matrix, 
or are identically zero on $\MM_d$. That is, we are interested in expressions 
$g: M_d^r\to M_d^{\ot t}$
with $r \geq t$ of the form
\begin{equation}
 g(X_1,\dots, X_r) = \sum_{i} p_i \ot \cdots \ot q_i\,,
\end{equation}
where $p_i, \dots,  q_i$ are non-commutative polynomials 
such that $g=0$, $g\propto \one$, or $g\propto T(\pi)$
for all $X_1,\dots, X_r \in \MM_d$.

It is known that all multilinear polynomial identities and central polynomials 
arise as a consequence of trace identities.
Here we show the analogous result for tensor polynomials.

\begin{theorem}\label{thm:swap}
Every multilinear swap polynomial is the consequence of some trace identity.
More generally, every multilinear permutation polynomial and tensor polynomial 
identity is the consequence of some trace identity.
\end{theorem}
\begin{proof}
We treat the case of swap polynomials first.
Let $g(X_1, \dots, X_{k-2})$ be a multilinear swap polynomial in $k-2$ variables in $\MM_d$. 
With some $\a \in \C S_k$, it can be written as
\begin{equation}
 g(X_1, \dots, X_{k-2}) = \tr_{1 \dots k \backslash \{k-1, k\}}[\hat T(\a) X_1 \ot \cdots \ot X_{k-2} \ot \one \ot \one]\,,
\end{equation} 
where each permutation $\pi$ appearing in $\a$ has a unique decomposition into two cycles $\pi = \sigma_1 \sigma_2$
such that $\sigma_1$ ($\sigma_2$) acts on position $k-1$ ($k$) non-trivially.
By assumption, the tensor polynomial $g$ is proportional to the swap $\Gamma$. 
It thus holds for any $A,B \in \MM_d$ that 
\begin{equation}\label{eq:swap_poly_proof_step}
 \tr[g(X_1, \dots, X_{k-2}) A \ot B] = c \tr\big[\Gamma (A \ot B)\big]\,.
\end{equation} 
for some $c = c(X_1, \dots, X_{k-2}) \in \C$. 
Note that the left-hand side of Eq.~\eqref{eq:swap_poly_proof_step} is multilinear and unitary invariant. 
So must be the right-hand side. 
In particular, this implies that $c(X_1, \dots, X_{k-2})$ is a unitary invariant, 
and thus a trace polynomial.
Therefore, the expression 
\begin{align}
  &\tr\Big\{
  \tr_{1 \dots k \backslash \{k-1, k\}} \big[\hat T(\a) X_1 \ot \cdots \ot X_{k-2} \ot \one \ot \one\big] A \ot B \Big\}
  - \tr\big [ c \Gamma A \ot B \big] \nn\\
  =
    &\tr\big[
\hat T(\a)  X_1 \ot \cdots \ot X_{k-2} \ot A \ot B
\big] 
- d^{-{(k-2)}} \tr\big[
c \Gamma \,\one \ot \cdots \ot \one \ot A \ot B
\big]
 \end{align}
 is a trace identity on $\MM_d$.
 By the Theorem of Procesi and Razmyslov~[Theorem~\ref{thm:ProcesiRazmyslov}], 
 this implies that $\big(\a - d^{-{(k-2)}} c\,(k-1,k)\big) \in \C S_k$ is in the ideal of the group algebra $\C S_k$ which is spanned by Young symmetrizers that have more than $d$~rows.
 We conclude that every swap polynomial arises from a trace identity.
 
 The general case for tensor polynomials and tensor trace polynomials is done analogously:
 replace Eq.~\eqref{eq:swap_poly_proof_step}
 with
 \begin{equation}
   \tr[g(X_1, \dots, X_{k-t}) A_1 \ot \dots \ot A_t] = c \tr\big[F (A_1 \ot \dots \ot A_t)\big]\,,
 \end{equation} 
 where $F$ is the desired operator and conclude that 
 $\hat T(\alpha) - cF$ corresponds to a trace identity.
 This ends the proof.
\end{proof}

For example, suppose that $\alpha \in \C S_k$ corresponds to a trace identity on 
$\MM_d$
and that every permutation appearing in $\alpha$ 
is composed of exactly two cycles $\pi = \sigma_1 \sigma_2$,
such that $\sigma_1$ ($\sigma_2$) acts on position $k-1$ $(k)$ non-trivially.
Then 
\begin{equation}\label{eq:swap_function}
h_\alpha ( X_1, \dots, X_{k-2}) = \tr_{1 \dots k \backslash\{k-1, k\}}
[\hat{T}(\beta)  X_1 \ot \dots \ot X_{k-2}  \ot \one \ot \one] 
\end{equation}
is identical to the zero matrix on $\MM_{d^2}$ 
whenever $X_1, \dots, X_{k-2} \in \MM_d$.

\begin{example}
Let $X_1,X_2,X_3,X_4$ be complex $2\times 2$ matrices. 
The following expression equals the $4\times 4$ zero matrix,
\begin{equation}\label{eq:tensorPolyID}
\sum_{\pi \in S_4} \sgn(\pi) X_{\pi(1)} X_{\pi(2)} \,\ot \,X_{\pi(3)} X_{\pi(4)} = 0\,.
\end{equation}
It is interesting to note that due to the theorem by Amitsur and Levitzki on the lowest degree of polynomial identities, 
the identity in Eq.~\eqref{eq:tensorPolyID} 
cannot be factorized into individual polynomial identities on the first and second tensor factors. 
\end{example}

We list further constructions in Appendix~\ref{app:construction_polyId}.
Tensor polynomial identities in $d^2$ variables are studied in Ref.~\cite{HuberProcesi}.

\section{Conclusions}
\label{sec:conclusion}
Relating to recent progress in the field of non-commutative polynomials, we presented a systematic
method to obtain polynomial-like matrix inequalities for the positive cone and 
identities on tensor product spaces.
In the field of quantum information, 
special cases of these maps, in particular that of $\lambda = [1,1]$, 
were useful on a range of topics.

Some natural questions remain unanswered.
First, do all tensor-stable multilinear equivariant positive 
maps arise from positive semidefinite operators,
or do there exist merely block-positive operators that can yield tensor-stability? 
This question is related to the existence of bound entangled states that have a 
non-positive partial transpose 
(c.f. Theorem~$3$ and $4$ in Ref.~\cite{doi:10.1063/1.4927070}),
a long-standing open problem in the quantum information community~\cite{KCIK2019}.

Second, we recall that Amitsur and Levitzki showed that $\MM_d$ satisfies the standard identity $s_{2d}$ in $2d$ variables while it does not satisfy any polynomial identity of lower degree. 
The corresponding question on the lowest degree for central polynomials, swap polynomials, 
and more general polynomials on tensor product spaces is unresolved~\cite{DrenskyFormanek2004}. 
Similarly, we do not know what is the lowest degree required when trace polynomials
instead of only polynomials are allowed. 
To approach this problem, a more direct construction for permutation polynomials 
and tensor polynomial identities than the ones presented here is desirable.

Third, in the context of quantum error correcting codes and maximally entangled states, 
it is interesting to understand for what conditions certain local unitary invariants vanish. 
This can be important to determine the existence of 
quantum error correcting codes and entangled subspaces 
with optimal parameters~\cite{
1751-8121-51-17-175301, 
Huber2020quantumcodesof}.

Last, it would be interesting to develop a general method to deal with non-linear (trace) polynomial inequalities, such as the Motzkin matrix inequality from Section~\ref{sect:Motzkin}.
An  approach that goes beyond the symmetrization of suitable multilinear inequalities
is likely required.

\begin{acknowledgments}
I thank 
Mohamed Barakat, 
Ben Dive, 
Mariami Gachechiladze,
Markus Grassl,
David Gross, 
Otfried Gühne,
Alexander Müller-Hermes, 
David Hill,
Barbara Kraus, 
Marcin Marciniak,
Chau Nguyen,
Michał Oszmaniec, 
Anna Sanpera,
Gael Sentís, 
Jens Siewert, 
Michał Studziński,
Jorge Urroz,
and 
Andreas Winter 
for fruitful discussions, pointers to relevant literature, and encouragement. 
I am indepted to members of the \code{GAP} group,
including
Max Horn, 
Alexander Hulpke, 
Chris Jefferson, 
Alexander Konovalov,
and
Markus Pfeiffer, 
for their kind technical help.
I thank 
Simeon Ball and 
Matteo Lostaglio 
for helpful discussions on the structure of this article;
Miguel Navascués 
for pointing out an error 
in the first version of this article;
Sébastian Designolle,
Hans Maassen, 
and 
Claudio Procesi 
for their careful reading and friendly feedback;
as well as
Jurij Volčič
for kindly communicating his proof of the Motzkin matrix inequality.

This work was supported by
the ERC (Consolidator Grant 683107/TempoQ), 
the DFG (GU 1061/4-1, SPP1798 CoSIP), 
the Excellence Initiative of the German Federal and State Governments (Grant ZUK 81), 
the Fundació Cellex,
the Spanish MINECO (QIBEQI FIS2016-80773-P and Severo Ochoa SEV-2015-0522), 
Generalitat de Catalunya (SGR 1381 and CERCA Programme), 
and
the European Union under Horizon2020 (PROBIST 754510).
\end{acknowledgments}

\appendix

\section{Representation theory of the symmetric group}\label{app:RepTheorySk}
The polarized Cayley-Hamilton map $f_\lambda$ 
is constructed from {\em central Young Projectors}.
These project onto the isotypic components of~$S_k$ and 
form a complete orthogonal set of hermitian projection operators 
that sum to the identity. While they can be obtained with character theory,
we show here how they can be constructed from centrally primitive idempotents 
in the group algebra $\C S_k$, starting from Young Tableaux alone.
We think that this makes the presentation more accessible for readers 
not versed in character theory. Also, it highlights the importance of 
ideals in $\C S_k$ for the construction of polynomial identities.
Nothing in this section is claimed to be original.

\subsection{Group algebra}
It is helpful to consider the irreducible representations of $S_k$ 
as arising from the group algebra~$\C S_k$. We will now introduce the formalism for a general group~$G$.

Given a group $G$, its {\em group algebra} $\C G$ is formed by the set of formal sums over group elements
\begin{equation}\label{eq:group_algebra_G}
 \C G = \big\{ \sum_{g \in S_k} a_g g \,|\, a_g \in \C \big\} \,.
\end{equation} 
It is not hard to see that $\C G$ is a vector space over $\C$.
The addition and the scalar multiplication are given by
\begin{align}
 \a + \b    &= \sum_{g \in G} a_g g + \sum_{g' \in G} b_{g'} g' 
            = \sum_{g \in G} (a_g + b_g)g 
    \qquad \text{where } \,\, \a,\b \in \C G\,,\nn\\
 c \cdot \a  &= c \cdot \big(\sum_{g \in G} a_g g \big) = \sum_{g \in G} c a_g g 
    \qquad\qquad\qquad\qquad\text{where }\,\, \quad \,\, c\in \C\,.
\end{align}
A multiplication is obtained by extending the multiplication on $G$ linearly,
\begin{align}
 \a \cdot \b = \Big(\sum_{g \in G} a_g g\Big) \Big(\sum_{g' \in G} b_{g'} g'\Big) 
 &= \sum_{g,g' \in G} a_g b_{g'} g g' = \sum_{g,g' \in G} a_{g g'}b_{g'^{-1}} g\,.
\end{align}
The group ring $\C G$ equipped with the above multiplication thus forms an algebra, 
also called the {\em left regular representation} of $\C G$. The left regular representation is {\em faithful} (injective).

Let now $\a = \sum_{g \in G} a_g g$ be an element of $\C G$ and define the involution $\a^* = \sum_{g \in G} \overbar{a}_g g^{-1}$.
If $\a^* = \a$ then $\a$ is termed {\em hermitian}.
If $\a$ commutes with all elements from the group algebra ($\a \b = \b \a$ for all $\b \in \C G$) it is called {\em central}. 
If $\a^2 = \a$ then $\a$ is {\em idempotent}. 
An idempotent is {\em primitive}, if $\a \neq \epsilon + \zeta$ such that $\epsilon\zeta = \zeta\epsilon = 0$ and $\epsilon, \zeta$ are idempotent.
In other words, $\a$ primitive if it cannot be written as a non-trivial sum of two orthogonal idempotents.
Similarly, an element $\a$ is {\em centrally primitive} if $\a \neq \epsilon + \zeta$ 
such that $\epsilon\zeta = \zeta\epsilon = 0$ and $\epsilon, \zeta$ are central and idempotent.
A {\em complete set} $\{\epsilon_i\}$ of orthogonal idempotents fulfills $\epsilon_i \epsilon_j = \epsilon_j \epsilon_i = 0$ for $i\neq j$ and $\sum_i \epsilon_i = e$,
where $e$ is the identity element in the group ring.

\subsection{Algebra representation}
A {\em representation} over $\C$ is a homomorphism $T: G \to \MM_d$ 
such that $T(g) T(h) = T(gh)$ for all $g, h \in G$. 
A representation $T$ is called {\em unitary}, 
if $\langle T(g) v, T(g) w\rangle = \langle v,w \rangle$ for all $v,w \in V$ and all $g$ in $G$.
The representation of $S_k$ that permutes individual tensor factors of $(\C^d)^{\ot k}$ 
as introduced in Sect.~\ref{sect:act_sym_group} is unitary.
If $T$ is a group representation of $G$, an {\em algebra representation} $\hat T$ is obtained by its linear extension,
\begin{equation}
 \hat{T}(\a) =\sum_{g \in G} a_g T(g)\,, \quad a_g \in \C\,.
\end{equation}
It is easy to check that $\hat{T}(\a \b) = \hat{T}(\a) \hat{T}(\b)$.
We have the following straightforward lemmata.
\begin{lemma}\label{lemma:central_hermitian}
    Let $G$ be a group with the property that every element is conjugate to its inverse. 
    Let $\a = \sum_{g \in G} a_g g$ with all $a_g \in \R$ be central. Then $\a$ is hermitian.
 \end{lemma}
\begin{proof}
 By linearity, $\a$ is central if and only if $\a h = h \a$ for all $h \in G$.
 In term of its coefficients, this condition is equivalent to $a_g = a_{h^{-1} g h}$ for all $g,h \in G$.
 But then $a_{g^{-1}} = a_{h' g h'^{-1}}$ for some $h'\in G$, because $g^{-1}$ is in the conjugancy class of~$g$.
 Thus also $a_{g^{-1}} = a_g$. This ends the proof. 
\end{proof}

\begin{lemma}\label{lemma:hermitian_hermitian}
Let $T$ be a unitary representation and let $\a \in \C G$ be hermitian. 
Then $\hat T(\a)$ is a hermitian matrix $\hat{T}(\a)^\dag = \hat{T}(\a)$.
\end{lemma}
\begin{proof}
 $
   \big(\hat{T}(\a)\big)^\dag = \big( \sum_{g \in G} \a_g T(g) \big)^\dag 
             = \sum_{g \in G} \overbar{\a}_g T(g)^\dag 
             = \sum_{g \in G} \overbar{\a}_g T(g^{-1}) 
            = \hat{T} \big( \sum_{g \in G} \overbar{\a}_g g^{-1} \big) 
             = \hat{T}(\a^*) = \hat{T}(\a)$.
This ends the proof.
\end{proof}
From the fact that $\hat T(\a^2) = \hat T(\a) \hat T(\a)$, it is now straightforward see that for unitary representations,
any hermitian idempotent $\a \in \C G$ will yield a hermitian projection operator
$\hat T(\a) = \hat T(\a)^\dag = \hat T(\a)^2$.

\subsection{Young symmetrizers}
Due to Maschke's Theorem, the group algebra $\C S_k$ can be decomposed completely into the direct sum of irreducible submodules~\cite{james_liebeck_2001}. 
An explicit construction can be obtained using Young tableaux; this works in the following way:
a {\em partition} $\lambda$ of an integer $k$ (written as $\lambda \vdash k$) is a sequence of 
positive integers $\lambda = (\lambda_1, \dots, \lambda_r)$, 
such that
\begin{equation}
 \lambda_1 \geq \lambda_2 \geq \dots \geq \lambda_r \quad \text{ and } \quad  \lambda_1 + \dots + \lambda_r = k\,.
\end{equation} 
A partition $\lambda$ can graphically be represented by its associated Young diagram. 
It consists of an arrangement of stacked squares 
such that $\lambda_i$ squares appear in the $i$-th row. 
Filling the numbers $1, 2, \dots, k$ with no repetitions into the squares, 
one arrives at a {\em Young tableau}.
For example, one could choose a natural way: starting with the top-most row, 
fill it with increasing numbers from the left to the right, 
before continuing with the row below.
In this fashion, the partition $\lambda = (4,3,1)$ leads to the tableaux
\begin{equation}\label{eq:young_tableaux_example}
\ytableausetup{centertableaux}
\TT = 
\begin{ytableau}
1 & 2 & 3 & 4 \\ 
5 & 6 & 7 \\ 
8
\end{ytableau}\,.
\end{equation}
A tableau is {\em standard} if the numbers increase along the rows and along the columns.

An element $c_{\lambda,\TT} \in \C S_k$ can be associated to each Young tableaux. 
Define the row- and column stabilizer~$\RR_{\TT}$ and~$\CC_{\TT}$ 
to consist of the set of permutations that leave the rows and columns respectively invariant.
\begin{align}
 \RR_{\TT} &= \{\pi \in S_k \,|\, \pi \text{ preserves each row} \} \,, \nn \\
 \CC_{\TT} &= \{\pi \in S_k \,|\, \pi \text{ preserves each column} \}\,.
\end{align}
Considering the partition tableaux in Eq.~\eqref{eq:young_tableaux_example}, 
the row-stabilizer $\RR_{\TT}$ is the subgroup $S_4 \times S_3 \times S_1$ of $S_8$,
where $S_4$, $S_3$, and $S_1$ act on $\{1,2,3,4\}$, $\{5,6,7\}$, and $\{8\}$ respectively.
Similarly, $\CC_{\TT} \cong S_3 \times S_2 \times S_2 \times S_1$.
Define the following elements in $\C S_k$ corresponding to the row and column stabilizer
\begin{align}
 a_{\lambda, \TT} &= \sum_{\pi \in \RR_{\TT}} \pi \,, &  b_{\lambda, \TT} &= \sum_{\pi \in \CC_{\TT}} \sgn(\pi) \pi \,.
\end{align}
The {\em Young symmetrizer} is then given by $c_{\lambda, \TT} = a_{\lambda, \TT} b_{\lambda,\TT}$.
The following can be shown~\cite{FultonHarris2004, Gill2005}:
some scalar multiple of $c_{\lambda,\TT}$ is idempotent, $c_{\lambda,\TT}^2= n_\lambda c_{\lambda,\TT}$ with $n_\lambda$ a positive scalar.
For all $x \in \C S_k$ one has that $c_{\lambda,\TT} x c_{\mu,T'} = 0$ if $\lambda$ and $\mu$ are partitions of $k$ with $\lambda \neq \mu$,
and $c_{\lambda,\TT} x c_{\lambda,\TT} = m c_{\lambda,\TT}$ with $m \in \R$.
Two Young symmetrizers whose tableaux $\TT$ and $\TT'$ have the same shape are related by 
$\sigma c_{\lambda, \TT} \sigma^{-1} = c_{\lambda, \TT'}$,
where $\sigma$ is the permutation for which $\TT' = \sigma(\TT)$.

Denote the subspace of~$\C S_k$ spanned by the Young symmetrizer $c_{\lambda,\TT}$ as
\begin{equation}
 V_{\lambda, \TT} = \{ c_{\lambda, \TT} x \, | \, x \in \C S_k \}\,.
\end{equation}
This is often written as $ V_{\lambda, \TT} = \C S_k c_{\lambda, \TT}$.
The following can be shown~\cite{FultonHarris2004}: 
the subspaces $V_{\lambda, \TT}$ are invariant under the action of $S_k$ and thus
each $V_{\lambda, \TT}$ leads to a irreducible representation of $S_k$. 
Subspaces originating from partitions of the same shape are isomorphic ($V_{\lambda,\TT} \cong V_{\lambda,\TT'}$) 
while those that arise from different partitions are not
(if $\lambda\neq \mu$ then $V_{\lambda,\TT} \ncong V_{\mu,\TT'}$).
All irreducible representations of $S_k$ arise from the $V_{\lambda, \TT}$.
Finally, the group algebra of the symmetric group decomposes into a direct sum of 
Young symmetrizers that correspond to partitions of different shapes, 
$\C S_k \simeq \bigoplus_{\lambda \vdash k} V_\lambda^{\oplus m_\lambda}$ with multiplicities 
$m_\lambda = \dim(V_\lambda)$ and $V_\lambda$ = $V_{\lambda,\TT}$ for an arbitrary tableau $\TT$.

We end our discussion on Young symmetrizers with the following well-known property.
\begin{proposition}[Pigeonhole principle for Young symmetrizers]
\label{prop:pigeon}
Let $\ket{\phi_1}, \dots, \ket{\phi_{k'}}$ be vectors in $\C^d$.
Suppose $\lambda$ is a partition of $k \leq k'$ that has more than $d$ parts.
Then for all tableaux $\TT$ of shape $\lambda$ the following holds,
\begin{equation}
    \hat{T}(c_{\lambda,\TT}) \ket{\phi_1} \ot \cdots \ot \ket{\phi_k} = 0\,.
\end{equation}
\end{proposition}
\begin{proof}
The proof rests on the pigeonhole principle.
First, note that the expression is multilinear and thus it suffices to prove 
it for vectors from an orthonormal basis only. 
Now observe that if at least $d+1$ pairs of vectors are anti-symmetrized, 
but only $d$ basis vectors are available, 
then at least one pair of anti-symmetrized vectors will coincide. 
Consequently the expression vanishes and this ends the proof.
\end{proof}

\subsection{Hermitian idempotents in $\C S_k$}
We first construct hermitian idempotents in $\C S_k$ which, 
under the representation $T$, yield the Young projectors $P_\lambda$.
Let $\lambda$ be a partition and let $c = c_{\lambda,\TT}$ be the Young symmetrizer 
that is obtained by filling the diagram from left to right and top to bottom with the 
numbers $1$ to $k$. 
We define the following element,
\begin{equation}\label{eq:omega_lambda}
    \omega_\lambda = \frac{h_\lambda}{k!}\sum_{\sigma \in S_k} \sigma c \sigma^{-1}\,.
\end{equation}
Above, the normalization factor involves the hook length formula
$ h_\lambda = k! / \prod_{(i,j)\in \lambda} h_{ij}$,
where the product is over all boxes $(i,j)$ indexed by row $i$ and column $j$ of the Tableau.
The hook length $h_{ij}$ equals the number of boxes that are below or to the right of box $(i,j)$, 
that is with $(i,j'\geq j)$ or $(i'\geq i, j)$ including box $(i,j)$.

The element $\omega_\lambda$ satisfies the following properties.
\begin{proposition}\label{prop:herm_idempotents}\ \\
 \noindent a) The elements $\omega_\lambda$ are mutually orthogonal, 
              $\omega_\lambda \omega_\mu = \omega_\mu \omega_\lambda  = 0$ 
              if $\lambda$ and $\mu$ are partitions of $k$ with $\lambda \neq \mu$.\\
 \noindent b) The elements $\omega_\lambda$ are central and hermitian,
              $\a \omega_\lambda = \omega_\lambda \a$ for all $\a \in \C S_k$ and $\omega_\lambda = \omega_\lambda^*$.  \\
 \noindent c) The elements $\omega_\lambda$ are non-vanishing idempotents,
              $\omega_\lambda^2 = \omega_\lambda \neq 0$.  \\
 \noindent d) The elements $\omega_\lambda$ form a complete set of idempotents, 
                $e = \sum_{\lambda \vdash k} \omega_\lambda$.
\end{proposition}
\begin{proof} We show everything except of the normalisation factor, which can be 
found in Ref.~\cite{Gill2005}.\\ 
    \noindent a): This follows directly from the fact that $c_{\lambda,\TT} c_{\mu,T'} = c_{\mu,T'} c_{\lambda,\TT} = 0$ if $\lambda \neq \mu$. \\
    \noindent b): 
                Write $\omega_\lambda = \sum_{\pi \in S_k} w_\pi \pi$.
                In term of its coefficients, $\omega_\lambda$ is central if and only if 
                $w_\pi = w_{\mu^{-1} \pi \mu}$ for all $\pi,\mu \in S_k$.
                By construction, 
                \begin{align}
                \omega_\lambda 
                &= \frac{h_\lambda}{k!}\sum_{\sigma \in S_k} \sigma c \sigma^{-1}
                = \frac{h_\lambda}{k!}\sum_{\sigma \in S_k} \sigma \Big(\sum_{\pi \in S_k} c_{\pi} \pi \Big) \sigma^{-1}
                = \frac{h_\lambda}{k!}\sum_{\sigma, \pi \in S_k} c_{\pi} \sigma \pi \sigma^{-1} 
                = \frac{h_\lambda}{k!}\sum_{\sigma, \pi \in S_k} c_{\sigma^{-1} \pi \sigma} \pi \,.
                \end{align} 
                Then the coefficients of $\omega_\lambda$ are given by $w_\pi = \sum_{\sigma \in S_k} c_{\sigma^{-1} \pi \sigma}$
                and it is easy to check that $w_\pi = w_{\sigma^{-1} \pi \sigma}$ for all $\pi,\sigma \in G$. 
                Thus $\omega_\lambda$ is central.
                Note now that any element in $S_k$ is conjugate to its inverse.
                It follows from Lemma~\ref{lemma:central_hermitian} 
                that $\omega_\lambda$ is also hermitian. \\
    \noindent c): It is clear that $\omega_\lambda$ does not vanish, because each $\sigma c \sigma^{-1}$ in the sum 
                  contributes with a factor $+1$ to the coefficient of~$w_e$. Thus $w_e\neq 0$ and consequently also $\omega_\lambda \neq 0$.\\                  
                  It remains to show idempotency: from the Artin-Wedderburn theorem it follows that any semi-simple ring over $\C$ is isomorphic 
                  to the direct sum of matrix rings~\cite[Theorem 2.1.3]{webb_2016}. 
                  Consequently, one has $\C S_k \cong \bigoplus_i M_{d_i}(\C)$, 
                  where $M_{d}(\C)$ denotes the ring of complex $d \times d$ matrices.
                  The center of $M_{d}(\C)$ consists of elements that are scalar multiples of the identity.
                  It follows that the center of $\bigoplus_i M_{d_i}(\C)$ consists of elements of the form 
                  $\oplus_i c_i\one_{d_i}$ with $c_i \in \C$. In $\C S_k$, this corresponds to elements of the form 
                  $\sum_{i} c_i \epsilon_i$ with $\epsilon_i$ central and $\epsilon_i^2 = \epsilon_i = \epsilon_i^\dag$.
                  Note now that all terms $\sigma c \sigma^{-1}$ in the sum Eq.~\eqref{eq:omega_lambda} 
                  have support in the same isotypic component.
                  The central element $\sum_{\sigma \in S_k} \sigma c \sigma^{-1}$ must therefore correspond 
                  to a scalar multiple of the identity in exactly one of the matrix rings.
                  We conclude that $\omega_\lambda$ is proportional to an idempotent.
                  \\
    \noindent d): It can be shown that in a ring with identity, 
                    complete sets of mutually orthogonal centrally primitive idempotents in 
                    $\epsilon_i \in \C S_k$ biject with the decomposition of the group ring $\C S_k$ into 
                    isotypic components $\C S_k \epsilon_i $~\cite[Proposition 3.6.1.]{webb_2016}, \cite{Bartel2017};
                    see also Section~\ref{sect:wedderburn}.
                    The set $\{\omega_\lambda \, | \, \lambda \vdash k\}$ accounts for every 
                    irreducible submodule and each $\omega_\lambda$ is central. 
                    We therefore cannot add any additional submodule.
                    Thus $\sum_{\lambda \vdash k} \omega_\lambda$ must necessarily decompose the identity~$e \in \C S_k$.
                    This ends the proof.
\end{proof}
Proofs for a) - d) that use character theory can also be found in Ref.~\cite{Gill2005}.
We summarize Proposition~\ref{prop:herm_idempotents} 
by concluding that $\{ \omega_\lambda | \lambda \vdash n\}$ 
forms a complete set of mutually orthogonal 
centrally primitive hermitian idempotents. 

\begin{remark}
 Eq.~\eqref{eq:omega_lambda} is not particularly useful when doing calculations by hand.
 We thus give here an alternative method:
 to obtain $\omega_\lambda$, it is in principle enough to write down the Young symmetrizer 
 for any single tableau of that shape.
 Note that for any permutation $\pi$ that appears in $c_{\lambda, \TT}$, 
 the sum over conjugates $\sum_{\sigma \in S_k} \sigma c_{\lambda, \TT} \sigma^{-1}$ will 
 produce all permutations $\pi'$ that have the same cycle structure as $\pi$. 
 Consequently, we only need to keep track of the net fraction of any given cycle type appearing in $c_{\lambda, \TT}$.
 For example, the tableau
 \begin{equation}
    \ytableausetup{centertableaux}
    \TT = 
    \begin{ytableau}
    1 & 2 \\ 
    3\\ 
    \end{ytableau}
    \end{equation}
    yields the Young symmetrizer $c_{\lambda, \TT} = ()  + (12) - (13) - (123)$. 
    We observe that the net fraction of $1$- and $2$-cycles vanishes, 
    and that $c_{\lambda, \TT}$ contains half of all possible $3$-cycles. 
    This lets us conclude that $\omega_\lambda \propto 2() - (123) - (132)$.

Alternatively, the following sequence of commands computes all $\omega_\lambda \in \C S_k$ 
in the computational discrete algebra package~\code{GAP}~\cite{GAP4} .
\texttt{
\\
\indent G := SymmetricGroup(k); \\
\indent KG:= GroupRing(Rationals, G); \\
\indent e := CentralIdempotentsOfAlgebra(KG);}
\end{remark}

\subsection{Young Projectors}
The {\em central Young projectors} decompose the identity matrix $\one_{dk}$ acting on $(\C^d)^{\otimes k}$ 
into a set of mutually orthogonal hermitian projectors, 
each of which corresponds to a distinct isotypic component associated to some partition $\lambda$ of $k$.
For this we consider the action of the symmetric group on $(\C^d)^{\otimes k}$ as described in Section~\ref{sect:act_sym_group}:
under the representation $T$, elements from $S_k$ permute the $k$ tensor factors.
Given $\omega_\lambda$ and the algebra representation~$\hat T$, define the associated {\em Young projector}
$ P_\lambda = \hat{T}(\omega_\lambda)$.
The following corollary is immediate and mirrors all properties of the $\omega_\lambda$ 
as established in Proposition~\ref{prop:herm_idempotents} in the previous section.
\begin{corollary}\ \\
 \noindent a) The elements $P_\lambda$ are mutually orthogonal, 
            $P_\lambda P_\mu = P_\mu P_\lambda = 0$ 
            if $\lambda$ and $\mu$ are partitions of $k$ with $\lambda \neq \mu$.\\
 \noindent b) The elements commute with swaps and are hermitian, 
              $\Gamma_{ij} P_\lambda = P_\lambda \Gamma_{ij}$ and $P_\lambda ^\dag = P_\lambda$.\\
 \noindent c) The elements $P_\lambda$ are nonvanishing projectors, $P_\lambda^2 = P_\lambda \neq 0$.\\
 \noindent d) The elements $P_\lambda$ form a decomposition of the identity, 
            \(\one = \sum_{\lambda \vdash k } P_\lambda\).
\end{corollary}
\begin{proof}
All the assertions are straightforward consequences of Observation~\ref{prop:herm_idempotents}.\\
\noindent a) $P_\lambda P_\mu = \hat{T}(\omega_\lambda) \hat{T}(\omega_\mu) = \hat{T}(\omega_\lambda \omega_\mu) = 0$ 
    if $\lambda \neq \mu$. \\
\noindent b) All $\omega_\lambda$ are central and thus
             $\Gamma_{ij} P_\lambda = \hat{T}(\gamma_{ij}) \hat{T}(\omega_\lambda) 
                = \hat{T}( \gamma_{ij} \omega_\lambda ) = \hat{T}( \omega_\lambda \gamma_{ij}) = 
                \hat{T}( \omega_\lambda) \hat{T}(\gamma_{ij}) = P_\lambda \Gamma_{ij}$.
             Because the representation $T$ is unitary and the elements $\omega_\lambda$ are hermitian, 
             it follows from Lemma~\ref{lemma:hermitian_hermitian} that the $P_\lambda$ are hermitian. \\
\noindent c) $P_\lambda^2 = \hat{T}(\omega_\lambda) \hat{T}(\omega_\lambda) = \hat{T}(\omega_\lambda^2) 
                = \hat{T}(\omega_\lambda) = P_\lambda$ and $P_\lambda$ is a projector.
             Because of $\omega_\lambda \neq 0$ it follows that $P_\lambda \neq 0$.\\
\noindent d) $\one = \hat{T}(e) = \hat{T}(\sum_{\lambda \vdash k} \omega_\lambda) 
                = \sum_{\lambda \vdash k} \hat{T}(\omega_\lambda) = \sum_{\lambda \vdash k} P_\lambda$.\\
                This ends the proof.
\end{proof}
We conclude that the central Young projectors $\{P_\lambda \, | \, \lambda \vdash k\}$
form a decomposition of identity matrix $\one_{dk}$ into a set of
mutually orthogonal projection operators that commute with both 
the action of $S_k$ and with the diagonal 
action of $GL(\C_d)$ on the tensor factors.

\subsection{Wedderburn decomposition}\label{sect:wedderburn}
Some further comments on the construction of the Young projectors are of interest.
Consider a vector space~$V$ over a field~$K$ on which a group~$G$ acts on; $V$ is said to be a {\em $KG$-module}.
A $KG$-{\em submodule} is a subspace $W \subseteq V$, such that 
$\alpha \cdot w \in W$ for all $\alpha$ in~$\C G$ and $w \in W$, or equivalently, $g \cdot w \in W$ for all $g$ in~$G$ and $w \in W$.
In other words, a submodule is a subspace that remains invariant under the group action of~$G$.
A submodule is {\em irreducible} or {\em simple}, if it does not contain any non-trivial submodules but itself.
There is a correspondence between representations of $G$ over $K$ 
and $KG$-submodules~\cite[Chpt. 4]{james_liebeck_2001} and
every irreducible representation is isomorphic to some irreducible submodule of~$\C G$.

Now let $G$ be a finite group. 
It follows from Maschke's Theorem~\cite[Chpt. 8]{james_liebeck_2001} 
that any $KG$-module with $K=\R$ or $K=\C$ can be decomposed into a direct sum of irreducible $KG$-submodules 
$V = \bigoplus_i W_i$; 
the vector space $V$ is said to be {\em completely reducible}. 
The left regular representation $\C G$ is a $KG$-module and consequently can be completely decomposed.
Grouped into components that consist of direct sums of mutually isomorphic irreducible submodules,
one obtains
\begin{align} \label{eq:left_reg_rep_decomp}
 \C G = \bigoplus_{\a} \big(E_\a^{(1)} \oplus \dots \oplus E_\a^{(m_\a)} \big) 
 \simeq \bigoplus_{\a} E_\a^{\oplus m_\a}
\end{align}
where $E_\a \simeq E_\a^{(i)}$ for all $1\leq i\leq m_\a$.
It can be shown that the decomposition of $\C G$ 
into the isotypic components $E_\a^{\oplus m_\a}$ is unique, and 
one has $m_\a = \dim(E_\a)$ with $\dim (\C G) = \sum_{\a} m_\a^2 = |G|$.

Given some idempotent $\epsilon$ one obtains the submodule $\C G \epsilon = \{ \a \epsilon \,|\, \a \in \C G\}$. 
Any complete set of mutually orthogonal idempotents $\{\epsilon_i\}$
corresponds to the decomposition of $\C G$ into the submodules $\C G \epsilon_i $. 
In particular, if $G$ is finite, then $\epsilon_i$ is a primitive idempotent if and only if
$\C G \epsilon_i $ is irreducible~\cite[Prop. 7.2.1]{webb_2016}. 
Thus the decomposition of $\C G$ into irreducible submodules bijects with 
a set of mutually orthogonal primitive idempotents.

Recall that an element $\a$ is called {\em centrally primitive} if $\a \neq \epsilon + \zeta$ 
such that $\epsilon\zeta = \zeta\epsilon = 0$ and $\epsilon, \zeta$ are central and idempotent.
Given a complete set of mutually orthogonal idempotents $\{\epsilon_\a\}$ that are centrally primitive, 
one directly obtains the decomposition of $\C G$ into its minimal two-sided ideals 
or isotypic components~\cite[Proposition 3.6.1]{webb_2016}. 
One has that $\C G \epsilon_\a = E_\a^{\oplus m_\a}$, 
and thus Eq.~\eqref{eq:left_reg_rep_decomp} can be written as
\begin{equation}
 \C G = \bigoplus_\a \C G \epsilon_\a \,.
\end{equation} 
Above is also known as the {\em Wedderburn decomposition} of $\C G$ where $\C G \epsilon_\a$ are the Wedderburn components.
We note that the centrally primitive idempotents can also be constructed using character theory~\cite[Theorem 3.6.2]{webb_2016}.
The elements $\omega_\alpha$ are then obtained from
  \begin{equation}
   \omega_\alpha = \frac{\dim V_\alpha}{|G|} \sum_{\pi \in S_k} \chi_\alpha(\pi^{-1}) \pi\,.
  \end{equation}
Thus in the case of the symmetric group, 
we could obtain the Young projectors 
also as
\begin{equation}
 P_\lambda = \hat T(\omega_\lambda) = \frac{\chi_\lambda(e)}{k!} \sum_{\pi \in S_k} \chi_\lambda(\pi^{-1}) T(\pi)\,.
\end{equation}

\section{Tables}\label{app:tables}
In Table~\ref{fig:tables} we list all non-trivial polarized Cayley-Hamilton maps $f_\lambda$ up to degree $k = 4$.
The $f_\lambda$ are positive on the positive cone [Theorem~\ref{thm:pCH_pos}] 
and vanish on complex $d\times d$ matrices whenever the partition has more than $d$ parts~[Theorem~\ref{thm:polCH}].
Furthermore, the $f_\lambda$ are equivariant under unitaries [Proposition~\ref{prop:pCH_covariance}], 
tensor stable [Theorem~\ref{thm:pCH_is_tensor_stable}], and completely copositive [Proposition~\ref{prop:pCH_copos}].

\LTcapwidth=0.6\textwidth
\begin{longtable}{@{} l @{\quad} l @{\quad\quad} l @{}}

\caption{All (non-trivial) polarized Cayley-Hamilton maps $f_\lambda$ up to degree $k=4$ 
that correspond to the isotypic components of the symmetric group.
\label{fig:tables}
}\\
k & partition & $f_\lambda$\\
\toprule\endfirsthead
k & partition & $f_\lambda$\\
\toprule\endhead
2
&$[1, 1]$&$ \tr(A){} -A$\\ \addlinespace[0.3mm]
\midrule
3
&$[2, 1]$&$ 2\tr(A)\tr(B){} -BA -AB$\\ \addlinespace[0.3mm]
&$[1, 1, 1]$&$ \tr(A)\tr(B){} -\tr(A)B -\tr(BA){} +BA +AB -\tr(B)A$\\ \addlinespace[0.3mm]
\midrule
4
&$[3, 1]$&$ 3\tr(A)\tr(B)\tr(C){} +\tr(A)\tr(B)C +\tr(A)\tr(CB){} +\tr(C)\tr(A)B +\tr(BA)\tr(C){} -\tr(BA)C$\\&&$
 -CBA -BAC -CAB +\tr(CA)\tr(B){} -\tr(CA)B -BCA$\\&&$
 -ABC +\tr(B)\tr(C)A -ACB -\tr(CB)A$\\ \addlinespace[0.3mm]
&$[2, 2]$&$ 2\tr(A)\tr(B)\tr(C){} -\tr(A)CB -\tr(A)BC +2\tr(BA)C -\tr(CBA){} -\tr(C)BA$\\&&$
 -\tr(BCA){} -\tr(B)CA +2\tr(CA)B -\tr(C)AB -\tr(B)AC +2\tr(CB)A$\\ \addlinespace[0.3mm]
&$[2, 1, 1]$&$ 3\tr(A)\tr(B)\tr(C){} -\tr(A)\tr(B)C -\tr(A)\tr(CB){} -\tr(C)\tr(A)B -\tr(BA)\tr(C){} -\tr(BA)C$\\&&$
 +CBA +BAC +CAB -\tr(CA)\tr(B){} -\tr(CA)B +BCA$\\&&$
 +ABC -\tr(B)\tr(C)A +ACB -\tr(CB)A$\\ \addlinespace[0.3mm]
&$[1, 1, 1, 1]$&$ \tr(A)\tr(B)\tr(C){} -\tr(A)\tr(B)C -\tr(A)\tr(CB){} +\tr(A)CB +\tr(A)BC -\tr(C)\tr(A)B$\\&&$
 -\tr(BA)\tr(C){} +\tr(BA)C +\tr(CBA){} -CBA -BAC +\tr(C)BA$\\&&$
 +\tr(BCA){} -CAB -\tr(CA)\tr(B){} +\tr(B)CA +\tr(CA)B -BCA$\\&&$
 -ABC +\tr(C)AB +\tr(B)AC -\tr(B)\tr(C)A -ACB +\tr(CB)A$\\ \addlinespace[0.3mm]
\bottomrule
\end{longtable}

More maps can be obtained from entanglement witnesses, as stated in Theorem~\ref{thm:witness_to_ineq}.
An example is the following: consider the partition $\lambda_{-} = (1,\dots,1) \vdash k$ that corresponds to the 
completely anti-symmetric subspace and define
\begin{equation}
 f_{\mathcal{W}_{-}}(X_1, \dots, X_{k-1}) = \prod_{i=1}^{k-1} \tr(X_i) \one - f_{\lambda_{-}}(X_1, \dots, X_{k-1})\,.
\end{equation}
The map $f_{\mathcal{W}_{-}}$ is positive and the following holds~[Corollary~\ref{cor:werner_antisym}]:
there exists some set of nonzero ${X}_i \in \MM_d^+$ such that 
$\lambda_{\text{min}}\big\{f_{\mathcal{W}_{-}}({X}_1, \dots, {X}_{k-1})\big\} = 0$.
Consequently, $f_{\mathcal{W}_{-}}$ is an optimal trace polynomial inequality for the positive cone,
and $f_{\mathcal{W}_{-}}$ corresponds to an optimal entanglement witness 
$\mathcal{W}_{-}$ for completely anti-symmetric Werner states~[Theorem~\ref{thm:witness_to_ineq}].

\section{Constructions for polynomial and tensor identities}\label{app:construction_polyId}
To complement Section~\ref{sect:PI}, 
we outline the construction of interesting non-commutative polynomials
on tensor product spaces.
Recall that a polynomial identity is a polynomial in several matrix variables that vanishes on 
the set of all $d\times d$ matrices $\MM_d$ for some $d$. 
A special case of this are weak polynomial identities which are only required to vanish on the subset of traceless matrices.
Central polynomials yield an element from the center $C(\MM_d)$;
in other words, they evaluate to a scalar multiple of the identity matrix~$\one$.
The exact relation between these types of polynomials is not completely understood.

Swap and permutation polynomials were introduced in the context of remote time manipulation in Ref.~\cite{Trillo2019}. 
These are tensor polynomials that yield a scalar multiple of~$\Gamma$ 
or some other~$T(\pi)$.
{Tensor polynomial identities} yield the zero matrix on a tensor product space.
These concepts can straightforwardly be extended to expressions containing traces.

\begin{definition}
Suppose the following relations are satisfied 
for all $X_1, \dots, X_r \in \MM_d$.
Then a non-commutative polynomial $p:\MM_d^r \to \MM_d$ is termed a
\begin{itemize}
 \item[(i)]   {\em polynomial identity}, if $p(X_1, \dots, X_r) = 0 \in \MM_d$. 
 \item[(ii)]  {\em weak polynomial identity}, if $p(X_1, \dots, X_r) = 0\in \MM_d$ 
 when $X_1, \dots, X_r$ are traceless.
 \item[(iii)] {\em central polynomial}, if $p(X_1, \dots, X_r) = c \one\in \MM_d$ with $c=c(X_1, \dots, X_r) \in \C$. 
 \end{itemize}
  \noindent A tensor polynomial $p:\MM_d^r \to \MM_d^{\otimes t}$ is termed a
  \begin{itemize}
 \item[(iv)]  {\em swap polynomial}, 
              if $ p(X_1, \dots, X_r) = c \Gamma \in \MM_d^{\ot 2}$
              and a {\em permutation polynomial}, 
              if $p(X_1, \dots, X_r) = c T(\pi) \in \MM_d^{\ot t}$
              with $c=c(X_1, \dots, X_r) \in \C$.
 \item[(v)]   {\em tensor polynomial identity}, if $p(X_1, \dots, X_r) = 0 \in 
                \MM_d^{\ot t}$.
\end{itemize}
\end{definition}

The following approach follows closely to that of Procesi in his seminal article on the invariant theory of matrices~\cite{PROCESI1976306}. 
Indeed the cases (i) and (iii) were already considered there; our construction for swap and permutation polynomials 
as well as tensor identities however seems to be new.
We consider expressions of the form 
\begin{equation}
 g_\alpha(X_1, \dots, X_k) 
 = \tr_{1\dots k \backslash k}\big[ \hat{T}(\alpha) X_1\ot \dots \ot X_k \big] 
 \quad \text{where} \quad \a = \sum_{\pi \in S_k} a_\pi \pi \in \C S_k\,.
\end{equation}
It is worth to understand the effect of certain permutations that can occur in $\a$.
Suppose $\pi$ consists of a single cycle of length $k$ that acts on all positions non-trivially,
$\pi = \sigma$. 
Then $g_\pi(X_1, \dots, X_k) = R_{\sigma^{-1}}$ [Proposition~\ref{prop:generalized_swap_identities}].
Suppose $\pi$ contains the cycle of length $1$
such that $\pi(i) = i$ for some position $i< k$.
Then $g_\pi(X_1, \dots, X_k)=0$ on traceless matrices [Corollary~\ref{cor:perm_and_ptrace}].
Suppose $\pi$ keeps the last position unmoved, $\pi(k) = k$.
Then $g_\pi(X_1, \dots, X_{k-1}, X_k = \one ) \propto \one$ [Corollary~\ref{cor:perm_and_ptrace}].

Let us denote by $V_\lambda(\C S_k)$ the isotypic components of $\C S_k$. 
Given some dimension $d$, we define 
\begin{equation}
 \JJ_d(\C S_k) 
 = \!\!\bigoplus_{\substack{\lambda \vdash k \\ |\text{parts}(\lambda)| > d}} \!\!V_\lambda(\C S_k)
\end{equation} 
where the sum is over all isotypic subspaces whose associated Young tableaux have strictly more than $d$ rows.
The idea is that $\JJ_d(\C S_k)$ is too anti-symmetric to support a vector space of dimension~$d$.
As in Theorem~\ref{thm:polCH}, one has for all $\a \in \JJ_d(\C S_k)$ that 
$\tr_{ 1\dots k \backslash k}[\hat{T}(\alpha) X_1\ot \dots \ot X_{k-1} \ot X_k]  = 0$.
A consequence is the following.
\begin{theorem} 
 Let $\a \in \JJ_d(\C S_k)$ and  $X_1,\dots, X_k \in \MM_d$. 
 Suppose that
\begin{itemize}
 \item[(i)]   
 every permutation in $\a$ consists of a single cycle that acts on all positions non-trivially.
 Then 
 $g_{\alpha}(X_1, \dots, X_k)$
 is a {polynomial identity}.
 
 \item[(ii)]
 $\a$ can be written as  $\a = \beta + \gamma$, such that:
 every permutation in $\beta$ consists of a single cycle that acts on all positions non-trivially and
 every permutation in $\gamma$
 contains a cycle of length $1$ that leaves some position $i< k$ unchanged.
 Then 
 $g_{\beta}(X_1, \dots, X_k)$
 is a {weak polynomial identity}. 
 
 \item[(iii)]
  $\a$ can be written as $\a = \beta + \gamma$, such that:
  every permutation in $\beta$ consists of a single cycle that acts on all positions non-trivially and 
  every permutation in $\gamma$ leaves position $k$ unchanged. Then 
 $g_{\beta}(X_1, \dots, X_{k-1}, \one)$
 is a {central polynomial}.
 \item[(iv)]
  $\a$ can be written as $\a = \beta + \gamma$, such that:
  every permutation appearing in $\beta$ is composed of exactly two cycles $\pi = \sigma_1 \sigma_2$ such that
  $\sigma_1$ ($\sigma_2$) acts on position $k-1$ $(k)$ non-trivially;
  each permutation in $\gamma$ contains the cycle $(k-1, k)$.
  Then 
  \begin{equation}
   h_\beta ( X_1, \dots, X_{k-2}) = \tr_{1 \dots k \backslash \{k-1, k\}}
   [\hat{T}(\beta)  X_1 \ot \dots \ot X_{k-2}  \ot \one \ot \one] 
  \end{equation}
  is a {swap polynomial}. The construction for permutation polynomials is analogous
  and for a permutation of length~$t$ one requires $\beta$ to be composed of exactly~$t$ cycles, each of which 
  act on a single position in $\{k-t+1, \dots, k\}$ non-trivially; 
  whereas each permutation in $\gamma$ contains the desired cycle on the positions $k-t+1, \dots, k$.
    Then 
  \begin{equation}
   h_\beta ( X_1, \dots, X_{k-t}) = \tr_{1 \dots k \backslash \{k-t+1,\dots, k\}}
   [\hat{T}(\beta)  X_1 \ot \dots \ot X_{k-t}  \ot \one \ot \dots \ot \one] 
  \end{equation}
  is a {permutation polynomial}. 
  
  \item[(v)]
   each permutation appearing in $\alpha$ is composed of exactly $t$ cycles $\pi = \sigma_1 \dots \sigma_t$ such that
  $\sigma_1, \dots, \sigma_t\) acts on position $k-t+1, \dots, k$ respectively non-trivially.
  Then 
  \begin{equation}
   h_\alpha ( X_1, \dots, X_{k-2}) = \tr_{1 \dots k \backslash \{k-1, k\}}
   [\hat{T}(\beta)  X_1 \ot \dots \ot X_{k-t}  \ot \one \ot \dots \ot \one] 
  \end{equation}
  is a {tensor polynomial identity}.
 \end{itemize}
\end{theorem} 
\begin{proof}
  Our preceding discussion establishes (i) - (v).
\end{proof}
Note that central polynomials and tensor polynomials 
that yield the identity matrix are special cases of permutation polynomials.
Here the permutations in $\a$ keep the last $t$ positions fixed.
Similar constructions as those above can be made to obtain 
expressions that are only valid on the set of traceless matrices ({\em weak} identities)
or that contain traces (tensor trace polynomials).

\FloatBarrier
\bibliography{current_bib}
\end{document}